\documentclass[3p]{elsarticle}
\biboptions{sort&compress}
\usepackage{amsmath}
\usepackage{amsfonts}
\usepackage{amsthm}
\usepackage{graphicx}
\usepackage{graphics}
\usepackage{color}
\usepackage{bbold}
\newcommand{\idmat}{\mbox{$1 \hspace{-1.3 mm} 1$}}
\newcommand{\R}{\mathbb{R}}
\newcommand{\be}{\begin{equation}}
\newcommand{\ee}{\end{equation}}

\newtheorem{trm}{Theorem}%[section]
\newtheorem{lem}[trm]{Lemma}
\newtheorem{prop}[trm]{Proposition}

\title{Single-cone real-space finite difference scheme for the time-dependent Dirac equation}
\author[KFU-Graz]{Ren\'{e} Hammer}
\ead{rene.hammer@uni-graz.at}
\author[KFU-Graz]{Walter P\"{o}tz\corref{cor1}}
\ead{walter.poetz@uni-graz.at}
\author[TU-Wien]{and Anton Arnold}
\ead{anton.arnold@tuwien.ac.at}
\address[KFU-Graz]{Institut f\"{u}r Physik, Karl-Franzens-Universit\"{a}t Graz, Universit\"{a}tsplatz 5, 8010 Graz, Austria}
\address[TU-Wien]{Institut f\"{u}r Analysis und Scientific Computing, TU-Wien, Wiedner Hauptstr. 8, 1040 Wien, Austria}
\cortext[cor1]{Corresponding author}

\begin{document}

\begin{abstract}
A finite difference scheme for the numerical treatment of the (3+1)D Dirac equation is presented.  Its staggered-grid intertwined  discretization treats space and time coordinates on equal footing, thereby avoiding the notorious fermion doubling problem. This explicit scheme  operates entirely in real space and leads to optimal linear scaling behavior for the computational effort per space-time grid-point.  It allows for an easy and efficient parallelization.  A functional for a norm on the grid is identified.  It can be interpreted as probability density and is proved to be conserved by the scheme. The single-cone dispersion relation is shown and exact  stability conditions are derived.  Finally, a single-cone scheme for the two-component (2+1)D Dirac equation, its properties, and a simulation of scattering at a Klein step  are presented.  
%neu???
\end{abstract}

\maketitle

\section{Introduction}

\subsection{The Dirac Equation and Numerical Schemes}

Even more than $80$ years after its presentation by P. A. M. Dirac in 1928, the Dirac equation has not lost its fascination and significance in physics \cite{dirac}. The Dirac equation (in quantized form) has been of fundamental importance to the development of modern field theories and many-particle physics \cite{sakurai,ryder,itzykson,srednicki}. 
 Interestingly, even today  analytic solutions are very rare\footnote{E.g., spin-$1/2$ particle in a homogeneous magnetic field \cite{rabi}, the Dirac oscillator \cite{nikolsky}, an electromagnetic plane wave \cite{wolkov}, or the Coulomb potential \cite{bethe}.} \cite{rabi,nikolsky,wolkov,bethe}. In general, solutions have to be obtained numerically.  They have gained in relevance by rapidly increasing computational resources, as well as the development of efficient numerical schemes. 
In a single particle picture, the (3+1)D Dirac equation is applicable whenever the external electromagnetic fields are strong enough to accelerate a spin-$1/2$ particle to relativistic speeds, but many-particle effects and electron-positron pair production can be neglected \cite{salamin,piazza}. This regime is reached in the study of light-matter interaction  with the availability of short-pulse laser light in the (sub-)femtosecond range and intensities in excess of $10^{18}$W/cm$^2$, which corresponds to the relativistic threshold\footnote{The electron is accelerated to relativistic speeds during one laser cycle.} \cite{salamin,piazza}. 
Much of the physics in strong laser fields has been understood within a classical treatment of the relativistic electron. More recently, a numerical treatment of the quantum wave packet dynamics has become feasible \cite{mocken,gourdeau}.  For an electron in a plane wave field, a wave-packet description of an electron reveals a drift of the wave packet in the direction of light propagation along with its spreading and shearing \cite{mocken}. For this investigation a (2+1)D FFT-split-operator code was used. In such an approach, the propagation induced by the momentum part of the Hamiltonian is computed in momentum space, and the remainder in real-space, using fast Fourier transformation between the two representations \cite{mocken}. The computational effort scales like $\mathcal{O}(N \ln N)$ where $N$ is the number of grid-points. An efficient code using operator splitting in real space, in combination with the exact characteristic solutions in each space direction, was introduced for the (3+1)D case of the 
Dirac equation recently \cite{gourdeau}.  It leads to the highly efficient operations count of $\mathcal{O}(N)$. 

In condensed matter physics, relativistic effects frequently are well accounted for by corrections to the Pauli equation derived from the Dirac equation \cite{sakurai}. 
%??? okh
Recently, however, (topological) condensed matter systems supporting effective Dirac and Majorana fermions have become a  vivid playground for this community \cite{neto,qi,hasan}. In particular, metallic surface states on topological insulator surfaces display 2D Dirac cone dispersion \cite{hsieh, hasan, xia, analytis}.  A dynamic analysis of such Dirac fermions, e.g., in presence of effective electromagnetic fields calls for numerical schemes which faithfully represent the low-energy excitation spectrum.  
We have recently developed and applied such schemes  for the (1+1)D and  (2+1)D effective 2-component Dirac equation.  In the (1+1)D case, we have presented a single-cone lattice scheme for which exactly absorbing open boundaries were derived \cite{hammer1D}.  In (2+1)D, we have first developed a staggered grid scheme with one additional artificial cone which, however, is able to preserve the linear dispersion of the free $m=0$ Dirac spin-$1/2$ particle along $x$- and $y$-axis \cite{hammer2D2cone}. More recently, a staggered grid single-cone scheme was developed for the two-component (2+1)D Dirac equation and used in studies of Dirac fermion dynamics on textured TI surfaces \cite{hammerDW, hammerAPL}. This scheme, its properties,  and its generalization to the four-spinor-component  (3+1)D Dirac equation are  topic of this paper.
It operates entirely in real space and, due to symmetric staggering,  treats space and time on equal terms. This scheme avoids the notorious fermion doubling characteristic of direct discretization of the Dirac derivative operators on a real space lattice \cite{nielsen}. To our knowledge, this is the first multi-D finite difference scheme with this property. This is achieved by redistributing the spinor components on a grid, staggered in space {\it and} time, such that individual spinor components sit on different (adjacent) time sheets.   The proposed scheme, as will be shown below, shows an $\mathcal{O}(N)$ scaling behavior.

\subsection{The Fermion Doubling Problem} 

Real-space finite-difference schemes for the Dirac equation have been plagued by the fermion doubling problem\footnote{It manifests itself in a non-monotonic dispersion relation leading to additional Dirac cones, in addition to the one at $k=0$ of the underlying continuum Hamiltonian. For $d$ discretized spatial dimensions one ends up with up to $2^d$ cones (``fermion flavors").}. It was shown rigorously by Nielsen and Ninomiya that the discretization of the Dirac equation on a regular space grid forbids a single chirally invariant fermion flavor without breaking one or more of the following assumptions: translational invariance, locality, and Hermiticity \cite{nielsen}.  Obtaining additional  spurious solutions due to discretization is not a problem specific to the Dirac equation, but can occur in all the cases where one discretizes a first derivative operator on a grid. Standard symmetric finite-difference approximations for the first derivatives are applied to preserve Hermiticity. They leave out the 
central grid point which, in turn, can take on arbitrary values without changing the specific value which the finite difference expression yields. For example, both a constant function and a function oscillating with the maximum frequency which can be resolved on the grid lead to zero for the central finite difference expression. 
%??? okh
Already the simplest model, the advection equation in 1D shows a non-monotonic behavior with a second minimum in the dispersion when a central approximation for the first spatial derivative is used (e.g., the forward-time central-space (FTCS) method).  In other words, a symmetric first derivative in space utilizes twice the lattice spacing of the underlying grid thereby, in the language of solid state physics, shrinking the effective Brillouin zone in this direction by a factor of two.

Fermion doubling on a grid can be avoided by following basically two main strategies: (i) the incorporation or (ii) the complete 
%??? okh
avoidance of the central grid point in the scheme. As to the former, the use of a one-sided finite difference operator leads to so-called upwind schemes. For the Dirac equation this seems to work only in 1D, when a unitary time-evolution (conserving the norm) is to be maintained \cite{stacey}.  In this particular case it is equivalent to distributing the spinor components over a staggered grid \cite{stacey, hammer1D}. Unfortunately in higher dimensions this strategy either leads to non-Hermiticity and a non-unitary scheme or again to the introduction of fermion doubling.
We note that the claim for a recently presented coordinate space operator splitting scheme was not troubled by fermion doubling could not be verified by us \cite{gourdeau}.  In contrast, we find that monotonic dispersion behavior is limited to $K=c\Delta_t/\Delta_x \leq 1/2$ for the (2+1)D and (3+1)D case within this scheme. However, in its present form it can only be executed for $K\in \mathbb{N}$, because characteristic solutions for advection equations, e.g. $f(x,t + \Delta_t) = f(x + K \Delta_x,t)$, are used. It might be possible to repair the scheme by using the reservoir technique, where the main idea is to wait $n$ time steps before updating the solution, which then allows one to use a smaller $\tilde{\Delta}_t = \Delta_t/n$ \cite{lorin}.

The presence of the central grid point in the scheme can also be enforced by adding artificial terms to the scheme, such as the momentum-dependent mass term suggested by Wilson, which lifts (splits) the spurious cones  at high momentum, but retains a $k=0$ cone in good approximation of the continuum dispersion \cite{wilson}. This strategy comes with its price.  In case of the Dirac equation,  it spoils chiral invariance  for the physical mass $m=0$ case\footnote{For $m\neq0$  chiral symmetry already 
%??? okh
is broken in the continuum problem.} \cite{wilson}.   Alternatively, one can violate one of the other premises of the Nielsen-Ninomiya no-go theorem, for example, by breaking translational invariance in an extra dimension. This leads to domain-wall fermions \cite{kaplan,nielsen}. Interestingly, nature uses this trick in topological insulators \cite{qi}. 

The other option is to retain the centered approximation and to ensure that the disturbing central grid point cannot cause a problem: get rid of the point -  get rid of the problem. This is achieved by the use of a staggered ("checkerboard")  grid  whereby the grid is subdivided into subsets, with a particular spinor component defined on one of them, but not on the other(s). For the Dirac equation the idea of a staggered grid may be tracked back to Kogut and Susskind in the context of lattice QCD \cite{susskind}.
Here we utilize a staggered grid in space and time which treats time and space on equal footing, in accordance with the covariance of the underlying continuum model. 
Given initial conditions on a time slice, time propagation is executed in a leap-frog manner where the upper spinor components are computed from the lower ones and vice versa in an explicit recursive scheme.  

The paper is organized as follows.  In Section \ref{sect2} we present the 
(3+1)D scheme and proceed with a discussion of its properties, such as norm conservation, stability, and the scaling behavior.  In Section \ref{sect3} we present this analysis for the corresponding scheme applied to (2+1)D. The gauge-invariant inclusion of an external electromagnetic potential and its consequences are outlined in Section \ref{sect4}. In Section \ref{sect5} we use the (2+1)D scheme to show the scattering of an initial Gaussian wave packet at a Klein step. Summary and conclusions are given in Section \ref{suco}.

%%%%%%%%%%%%%%%%%%%%%%%%%%%%%%%%%%%%%%%%%%%%%%%%%%%%%%%%%%%%%%%%%%%%%%%%%%%%%%%

\section{Numerical Scheme for the (3+1)D Dirac Equation}\label{sect2}
\begin{figure}[t!]
\centering
\includegraphics[width=17cm]{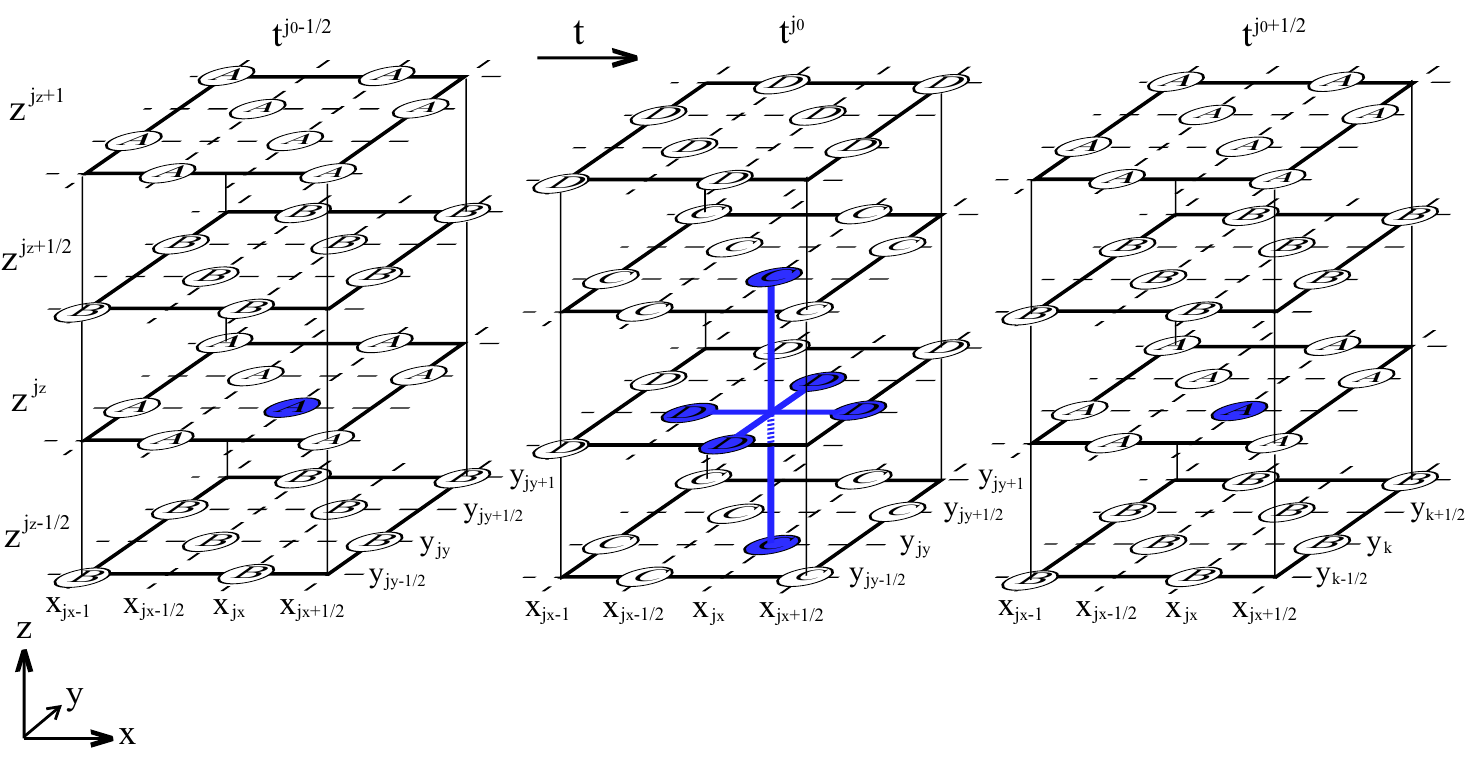}
%????neu
\caption{(color online). Space-time stencil for the first sub-step with explicit leap-frog time stepping for the space-time staggered finite difference scheme of the (3+1)D Dirac equation: The spinor components $A$, $B$, $C$, and $D$ are defined on special positions of the space-time grid in order to allow a centered approximation of all first derivatives without inducing fermion doubling.  The time propagation is sketched  in horizontal direction.  The z-axis is given the vertical direction.  In the first step the two spinor components $A$ and $B$, initially placed on time sheet $t^{j_o-1/2}$, are propagated to time sheet  $t^{j_o+1/2}$, using components $C$ and $D$ located on time sheet $t^{j_o}$.  For component $A$,  the spinor components entering the centered symmetric difference quotients representing the partial derivatives  are indicated in dark (blue) color.  In the second half-step components $C$ and $D$ are propagated  from time sheet $t^{j_o}$ to time sheet  $t^{j_o+1}$ (not shown in the figure) 
using centered differences of components $A$ and  $B$ on time sheet $t^{j_o+1/2}$.  }
%????neu
\label{scheme3D}
\end{figure}
\noindent The single particle Dirac equation offers a
%??? okh
relativistic description of spin $1/2$ particles capturing their particle-antiparticle and spin degree of freedom. A common representation is the  (3+1)D Schr\"{o}dinger form  \cite{greiner,sakurai}

\begin{equation}
i\hbar c \frac{\partial}{\partial x_o}{\vec \psi}(x_\mu)= \left[-i\hbar c \sum_{j=1,2,3} \alpha_j  \frac{\partial}{\partial x_j} + mc^2 \beta + V(x_\mu)\idmat_4 \right] {\vec \psi}(x_\mu)~.
\label{3+3D}
\end{equation}
Here we use the $4 \times 4$ Dirac matrices in Pauli-Dirac (``standard") form \cite{greiner}
\begin{equation}
\alpha_j  = \sigma_x\otimes \sigma_j = \left(\begin{array}{cc} 0 & \sigma_j \\ \sigma_j & 0 \end{array}\right)  
\end{equation}
and 
\begin{equation}
\beta  = \sigma_z\otimes \idmat_2 = \left(\begin{array}{cc} \idmat_2 & 0\\ 0 & -\idmat_2 \end{array}\right) ~,
\end{equation}
expressed by the Pauli matrices $\sigma_j$ and the $2\times2$ unit matrices $\idmat_2$. The first element in the direct product form accounts for the particle-antiparticle subspace and the second one represents the spin degree of freedom.  As  usual, $(x_\mu)=(x_o=ct, x_1=x,x_2=y,x_3=z)$ is the  space-time four vector 
%????neu
in Minkowski space.  
\\
Writing the Dirac spinor as
\begin{equation}
 {\vec \psi}(x_\mu) = \left(\begin{array}{c} \psi_1(x_\mu) \\ \psi_2(x_\mu)\\ \psi_3(x_\mu)\\ \psi_4(x_\mu) \end{array}\right)  =\left(\begin{array}{c} A(x_\mu) \\ B(x_\mu)\\ C(x_\mu)\\ D(x_\mu) \end{array}\right) 
\end{equation}
%????neu
 one observes that, in the standard representation, components $A$ and $B$, as well as $C$ and $D$, decouple.  The remaining couplings are of two types: the mass and potential terms couple the time derivative of a particular component with itself, secondly, 
 the time derivative of each of the upper two components couples to spatial derivatives of the lower two components.
%AAAneu
As an illustration we now write out the first line of \eqref{3+3D}:
\begin{equation}
i\hbar c \frac{\partial A}{\partial x_o}= 
-i\hbar c \left[\frac{\partial D}{\partial x_1} -i \frac{\partial D}{\partial x_2} 
+ \frac{\partial C}{\partial x_3} \right] + (m^2c+V)\,A~.
\label{Dirac-A}
\end{equation}
%neuAAA

 With $j_o$, $j_x$, $j_y$, and $j_z$ $\in\mathbb{Z}$, we define the time step progression in units  $\Delta_t$, with $\Delta_o=c\Delta_t$,  for the spinor in two steps as shown in Fig. \ref{scheme3D}.  First $A$ and $B$ are propagated from time-sheet $t-\Delta_t/2$ (i.e. index $j_o-1/2$)  to $t+\Delta_t/2$ (i.e. index $j_o+1/2$), followed by the propagation of $C$ and $D$ from time-sheet $t$ ($j_o$) to $t+\Delta_t$ ($j_o+1$).  One needs $x$-,~$y$-, and $z$-derivatives of $C$ and $D$ on time sheet $j_0$, for the former, and of $A$ and $B$ on time sheet $j_o+1/2$, for the latter update.  
 Each spinor component is placed  onto every other time sheet:  $A$ and $B$ are put on half-integer time sheets $j_o\pm1/2$ and $C$ and $D$ are put on integer time sheets, such as  $j_o$ and $j_o+1$.  This allows one to compute symmetric time derivatives with step width $\Delta_t$. If one then puts $A$ on a rectangular spatial grid, $j_x,j_y,j_z$ the implementation of symmetric $x$- and $y$-derivatives of $D$ requires the latter to be put on a face-centered rectangular (fcr) sublattice with grid points $(j_x,j_y+1/2)$, $(j_x+1/2,j_y)$ in the $x$-$y$ planes.  In order to be able, in turn, to provide symmetric $x$-and $y$-derivatives of $A$ when $D$ is updated, the sublattice of $A$ has to be extended from simple rectangular to fcr (adding lattice points $j_x+1/2,j_y+1/2$).
Since it is necessary to compute the $z$-derivative of $C$ to obtain $A$ at the new time sheet, $C$ has to be given at the $j_z+1/2$ sites, sharing its $x$- and $y$-positions with $A$. $A$ and $D$ may share the same $j_z$ grid points. Similarly, $B$ and $D$ share grid points on the $x$- and $y$- axis, while $B$ and $C$ share their grid points along the $z$-axis.  This determines all lattice sites on which individual spinor components need to be defined 
%??? okh
as follows (again with $j_o$, $j_x$, $j_y$ and $j_z$ $\in\mathbb{Z}$)
 \begin{align}
A(x_\mu)& ~\rightarrow~  A^{j_o-1/2,j_z}_{j_x,j_y}, \quad A^{j_o-1/2,j_z}_{j_x+1/2,j_y+1/2}~, \nonumber\\
B(x_\mu)& ~\rightarrow~  B^{j_o- 1/2,j_z+1/2}_{j_x+1/2,j_y}, \quad B^{j_o- 1/2,j_z+1/2}_{j_x,j_y+1/2}~,  \nonumber\\
C(x_\mu)& ~\rightarrow~  C^{j_o,j_z+1/2}_{j_x,j_y}, \quad C^{j_o,j_z+1/2}_{j_x+1/2,j_y+1/2}~, \nonumber\\
D(x_\mu)& ~\rightarrow~  D^{j_o,j_z}_{j_x+1/2,j_y}, \quad  D^{j_o,j_z}_{j_x,j_y+1/2} ~~.
\end{align}

As the attentive reader probably has already recognized, the  construction of the grid above was influenced by the use of the standard representation and the requirement of symmetric derivatives using single-lattice constant discretization with a fully symmetric structure in time and space.   One can summarize: $A$ and $B$ ($C$ and $D$) live on the same time sheets, $A$ and $D$ ($B$ and $C$) on the same $z$-sheets, and $A$ and $C$ ($B$ and $D$) are defined on the same $x-y$ positions respectively.  
Let us discuss the number of grid-points available (occupied and unoccupied) vs. the number of grid-points where a spinor component is actually defined. The number of grid-points altogether (occupied and unoccupied) is $2^4=16$ times the number of grid-points available on a space-time-grid without the half-integer positions. Each component is defined on two times the number of integer positions. Thus we can conclude that $1/8$ of all available grid-points are occupied by one concrete spinor component (e.g. $A$). Remember that for an non-staggered grid, in three spatial dimensions, one obtains $2^3=8$ Dirac cones. Here, as we will see below \ref{sectdispersion}, one has only one Dirac cone.
Due to the staggering of the lattice (here also in time!), all partial derivatives can be performed in centered fashion about the grid points without modification of the "lattice constants".  

With $j_o$, $j_x$, $j_y$, and $j_z$ $\in\mathbb{Z}$ the scheme looks as follows.   For the update of the first spinor component $A$ on the $(j_x,j_y)$ integer grid-points one has (cp.~\eqref{Dirac-A})
\begin{align} 
\frac{A^{j_o+1/2,j_z}_{j_x,j_y} -A^{j_o-1/2,j_z}_{j_x,j_y}}{\Delta_o} +  \frac{C^{j_o,j_z+1/2}_{j_x,j_y} -C^{j_o,j_z-1/2}_{j_x,j_y}}{\Delta_z} 
 + \frac{D^{j_o,j_z}_{j_x+1/2,j_y} -D^{j_o,j_z}_{j_x-1/2,j_y}}{\Delta_x}  \nonumber \\ 
 - i \frac{D^{j_o,j_z}_{j_x,j_y+1/2} -D^{j_o,j_z}_{j_x,j_y-1/2}}{\Delta_y}  
 = \frac{1}{i\hbar c } ( mc^2 + V)^{j_o,j_z}_{j_x,j_y} \frac{A^{j_o+1/2,j_z}_{j_x,j_y} + A^{j_o-1/2,j_z}_{j_x,j_y}}{2} ~.\label{sbegin} 
\end{align}
For the half-integer $(j_x+1/2,j_y+1/2)$ grid-points the update is
\begin{align} 
\frac{A^{j_o+1/2,j_z}_{j_x+1/2,j_y+1/2} -A^{j_o-1/2,j_z}_{j_x+1/2,j_y+1/2}}{\Delta_o} +  \frac{C^{j_o,j_z+1/2}_{j_x+1/2,j_y+1/2} -C^{j_o,j_z-1/2}_{j_x+1/2,j_y+1/2}}{\Delta_z} 
 + \frac{D^{j_o,j_z}_{j_x+1,j_y+1/2} -D^{j_o,j_z}_{j_x,j_y+1/2}}{\Delta_x}  \nonumber \\ 
 - i \frac{D^{j_o,j_z}_{j_x+1/2,j_y+1} -D^{j_o,j_z}_{j_x+1/2,j_y}}{\Delta_y}  
 = \frac{1}{i\hbar c } ( mc^2 + V)^{j_o,j_z}_{j_x+1/2,j_y+1/2} \frac{A^{j_o+1/2,j_z}_{j_x+1/2,j_y+1/2} + A^{j_o-1/2,j_z}_{j_x+1/2,j_y+1/2}}{2}  ~. 
\end{align}
\\
For $B$ one has on the $(j_x+1/2,j_y)$ sub-grid
 \begin{align} 
\frac{B^{j_o+1/2,j_z}_{j_x+1/2,j_y} -B^{j_o-1/2,j_z}_{j_x+1/2,j_y}}{\Delta_o} -  \frac{D^{j_o,j_z+1/2}_{j_x+1/2,j_y} -D^{j_o,j_z-1/2}_{j_x+1/2,j_y}}{\Delta_z} 
 +\frac{C^{j_o,j_z}_{j_x+1,j_y} -C^{j_o,j_z}_{j_x,j_y}}{\Delta_x} \nonumber \\
 + i \frac{C^{j_o,j_z}_{j_x+1/2,j_y+1/2} -C^{j_o,j_z}_{j_x+1/2,j_y-1/2}}{\Delta_y} 
 = \frac{1}{i\hbar c } ( mc^2 + V)^{j_o,j_z}_{j_x+1/2,j_y} \frac{B^{j_o+1/2,j_z}_{j_x+1/2,j_y} + B^{j_o-1/2,j_z}_{j_x+1/2,j_y}}{2} ~,
 \end{align}
and on the $(j_x,j_y+1/2)$ sub-grid one uses
\begin{align} 
\frac{B^{j_o+1/2,j_z}_{j_x,j_y+1/2} -B^{j_o-1/2,j_z}_{j_x,j_y+1/2}}{\Delta_o} -  \frac{D^{j_o,j_z+1/2}_{j_x,j_y+1/2} -D^{j_o,j_z-1/2}_{j_x,j_y+1/2}}{\Delta_z} 
 +\frac{C^{j_o,j_z}_{j_x+1/2,j_y+1/2} -C^{j_o,j_z}_{j_x-1/2,j_y+1/2}}{\Delta_x} \nonumber \\
 + i \frac{C^{j_o,j_z}_{j_x,j_y+1} -C^{j_o,j_z}_{j_x,j_y}}{\Delta_y} 
 = \frac{1}{i\hbar c } ( mc^2 + V)^{j_o,j_z}_{j_x,j_y+1/2} \frac{B^{j_o+1/2,j_z}_{j_x,j_y+1/2} + B^{j_o-1/2,j_z}_{j_x,j_y+1/2}}{2} ~.
 \end{align}
\\
The update for the $C$-components living on $(j_x,j_y)$ writes
\begin{align} 
\frac{C^{j_o+1,j_z}_{j_x,j_y} -C^{j_o,j_z}_{j_x,j_y}}{\Delta_o} +  \frac{A^{j_o+1/2,j_z+1/2}_{j_x,j_y} -A^{j_o+1/2,j_z-1/2}_{j_x,j_y}}{\Delta_z} 
 + \frac{B^{j_o+1/2,j_z}_{j_x+1/2,j_y} -B^{j_o+1/2,j_z}_{j_x-1/2,j_y}}{\Delta_x}  \nonumber \\ 
- i \frac{B^{j_o+1/2,j_z}_{j_x,j_y+1/2} -B^{j_o+1/2,j_z}_{j_x,j_y-1/2}}{\Delta_y} 
 = \frac{1}{i\hbar c } ( -mc^2 + V)^{j_o+1/2,j_z}_{j_x,j_y} \frac{C^{j_o+1,j_z}_{j_x,j_y} + C^{j_o,j_z}_{j_x,j_y}}{2} ~, 
 \end{align}
and for the $C$-component on $(j_x+1/2,j_y+1/2)$
\begin{align} 
\frac{C^{j_o+1,j_z}_{j_x+1/2,j_y+1/2} -C^{j_o,j_z}_{j_x+1/2,j_y+1/2}}{\Delta_o} +  \frac{A^{j_o+1/2,j_z+1/2}_{j_x+1/2,j_y+1/2} -A^{j_o+1/2,j_z-1/2}_{j_x+1/2,j_y+1/2}}{\Delta_z} 
 + \frac{B^{j_o+1/2,j_z}_{j_x+1,j_y+1/2} -B^{j_o+1/2,j_z}_{j_x,j_y+1/2}}{\Delta_x}  \nonumber \\ 
- i \frac{B^{j_o+1/2,j_z}_{j_x+1/2,j_y+1} -B^{j_o+1/2,j_z}_{j_x+1/2,j_y}}{\Delta_y} 
 = \frac{1}{i\hbar c } ( -mc^2 + V)^{j_o+1/2,j_z}_{j_x+1/2,j_y+1/2} \frac{C^{j_o+1,j_z}_{j_x+1/2,j_y+1/2} + C^{j_o,j_z}_{j_x+1/2,j_y+1/2}}{2} ~. 
 \end{align}
\\
Finally, for the $D$ component on the $(j_x+1/2,j_y)$ sub-grid the update is
 \begin{align} 
\frac{D^{j_o+1,j_z}_{j_x+1/2,j_y} -D^{j_o,j_z}_{j_x+1/2,j_y}}{\Delta_o} -  \frac{B^{j_o+1/2,j_z+1/2}_{j_x+1/2,j_y} -B^{j_o+1/2,j_z-1/2}_{j_x+1/2,j_y}}{\Delta_z} 
 + \frac{A^{j_o+1/2,j_z}_{j_x+1,j_y} -A^{j_o+1/2,j_z}_{j_x,j_y}}{\Delta_x}  \nonumber \\ 
+ i \frac{A^{j_o+1/2,j_z}_{j_x+1/2,j_y+1/2} -A^{j_o+1/2,j_z}_{j_x+1/2,j_y-1/2}}{\Delta_y}   = \frac{1}{i\hbar c } ( -mc^2 + V)^{j_o+1/2,j_z}_{j_x+1/2,j_y} \frac{D^{j_o+1,j_z}_{j_x+1/2,j_y} + D^{j_o,j_z}_{j_x+1/2,j_y}}{2}  ~, 
 \end{align}
and, for $(j_x,j_y+1/2)$, 
 \begin{align} 
\frac{D^{j_o+1,j_z}_{j_x,j_y+1/2} -D^{j_o,j_z}_{j_x,j_y+1/2}}{\Delta_o} -  \frac{B^{j_o+1/2,j_z+1/2}_{j_x,j_y+1/2} -B^{j_o+1/2,j_z-1/2}_{j_x,j_y+1/2}}{\Delta_z} 
 + \frac{A^{j_o+1/2,j_z}_{j_x+1/2,j_y+1/2} -A^{j_o+1/2,j_z}_{j_x-1/2,j_y+1/2}}{\Delta_x}  \nonumber \\ 
+ i \frac{A^{j_o+1/2,j_z}_{j_x,j_y+1} -A^{j_o+1/2,j_z}_{j_x,j_y}}{\Delta_y}   = \frac{1}{i\hbar c } ( -mc^2 + V)^{j_o+1/2,j_z}_{j_x,j_y+1/2} \frac{D^{j_o+1,j_z}_{j_x,j_y+1/2} + D^{j_o,j_z}_{j_x,j_y+1/2}}{2}  ~.\label{send} 
 \end{align}
 
%????neu
A more compact representation of the scheme, in analogy to the continuum form of the Dirac equation \eqref{3+3D} in standard representation,  can be given for the spinor components ${\psi}_k(x_\mu^{(k)})$ as follows
\begin{equation}
i\hbar c D_o{\psi}_k(x_\mu^{(k)})= \sum_{l=1}^{4} \left[\frac{\hbar c}{i} \sum_{j=1,2,3} \alpha_j  D_ j + \left(mc^2  \beta + V(x_\mu^{(k)})\idmat_4 \right) T \right] _{k l}  \psi_l (x_\mu^{(k)})~; ~ x_\mu^{(k)}\in G_k, k=1,2,3,4.
\label{3+3Ddis}
\end{equation}
 Here we use the symmetric difference operators
 $$
 D_ \nu  \psi(x_\mu)= \frac{1}{\Delta_\nu}\left[ \psi(x_\mu+\frac{\Delta_\nu}{2}{\hat e}_\nu)- \psi(x_\mu-\frac{\Delta_\nu}{2}{\hat e}_\nu )\right]; ~ \nu=o,x,y,z
 $$
 and the time-average operator
 $$
 T  \psi(x_\mu)= \frac{1}{2}\left[ \psi(x_\mu+\frac{\Delta_o}{2}{\hat e}_o )+ \psi(x_\mu-\frac{\Delta_o}{2}{\hat e}_o )\right] .
 $$
 ${\hat e}_\nu$ denotes the unit vector along the $\nu$-axis in Minkowski space.  
 The grids $G_k$ are defined as follows
  \begin{align}
G_1: &  ~ (j_o,j_x,j_y,j_z) ~\cup~ (j_o,j_x+1/2,j_y+1/2,j_z)~, \nonumber\\
G_2:  & ~ (j_o,j_x+1/2,j_y,j_z+1/2) ~\cup~ (j_o,j_x,j_y+1/2,j_z+1/2)~,  \nonumber\\
G_3: &   ~ (j_o+1/2,j_x,j_y,j_z+1/2) ~\cup~  (j_o+1/2,j_x+1/2,j_y+1/2,j_z+1/2) ~, \nonumber\\
G_4: &  ~  (j_o+1/2,j_x+1/2,j_y,j_z) ~\cup~ (j_o+1/2,j_x,j_y+1/2,j_z)~, 
\end{align}  
for integer $j_\nu$.

%????neu

 \subsection{Dispersion relation}\label{sectdispersion}
\begin{figure}[t!]
\centering
\includegraphics[width=15cm]{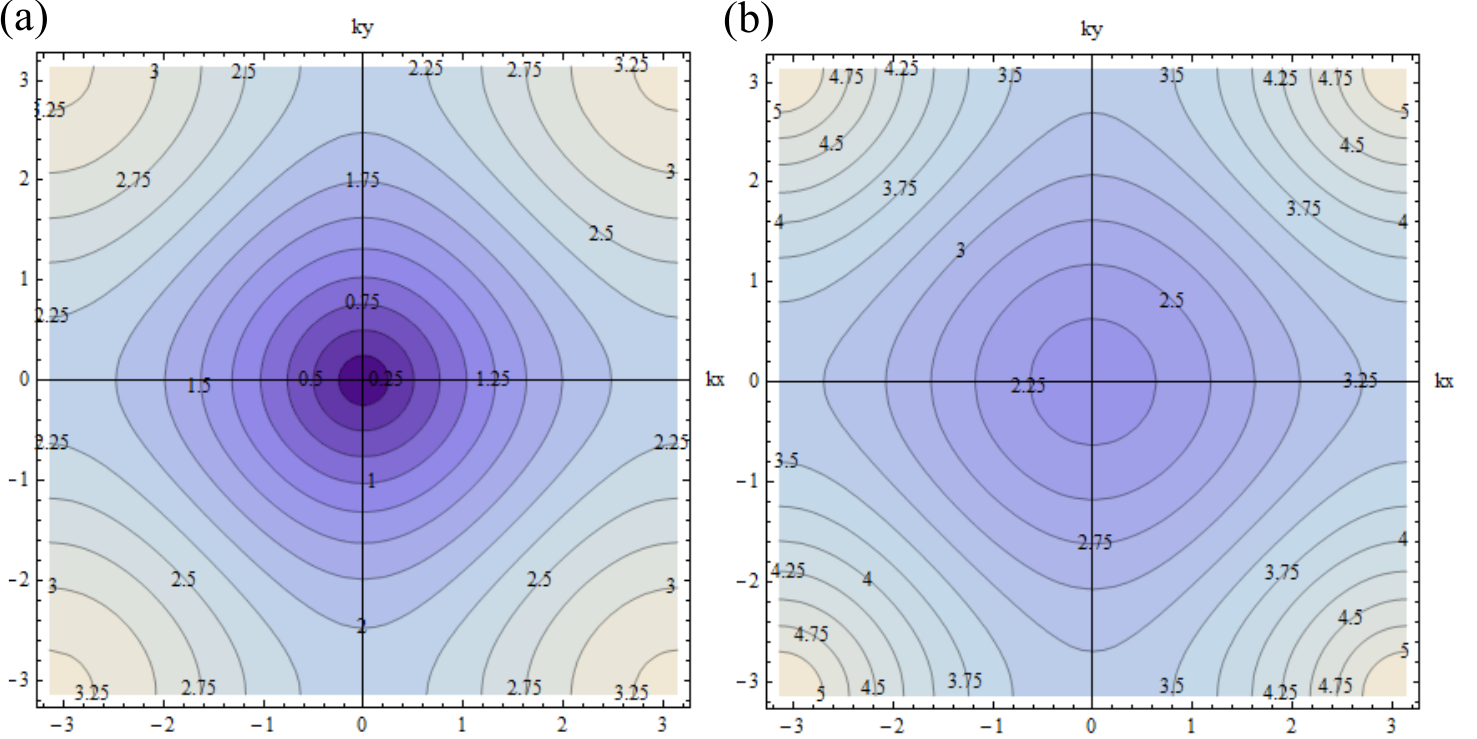}
\caption{(color online). Dispersion of the leap-frog staggered-grid finite difference scheme for the (3+1)D Dirac equation. $\Delta_x = \Delta_y = \Delta_z = \hbar = 1$, $\Delta_o = 1/\sqrt{3}$, $m=0$. (a) $k_z=0$, (b) $k_z=\pi$.}
\label{dispersion3D}
\end{figure}
\noindent There is no fermion doubling for this scheme.  For constant coefficients (potential and mass), this  can be shown by considering the equivalent to continuum plane-wave solutions for frequency $\omega$ and $k$-vector $(k_x,k_y,k_z)$,     
\begin{equation}
 {\vec \psi}_{k_\mu} (x_\mu) = \left(\begin{array}{c} \tilde{A}(k_j) \\ \tilde{B}(k_j)\\ \tilde{C}(k_j)\\ \tilde{D}(k_j)\end{array}\right)  e^{-i\omega t} e^{i \sum_j k_j x_j}~.
\end{equation}
On the lattice, given by the scheme, plane-wave solutions  take the form
\begin{equation}
 {\vec \psi}_{k_\mu} (x_\mu) = \left(\begin{array}{c} \tilde{A}(k_j) e^{i \omega \Delta_t/2} \\ \tilde{B}(k_j) e^{i\omega \Delta_t/2} e^{i k_x \Delta_x/2} e^{i k_z \Delta_z/2}\\ \tilde{C}(k_j) e^{i k_z \Delta_z/2}\\ \tilde{D}(k_j) e^{i k_x \Delta_x/2}\end{array}\right)  e^{-i \omega j_t \Delta_t} e^{i \sum_j k_j j_j \Delta_j} ~,\label{planeonlattice}
\end{equation}
where $\omega = c k_o$, $\Delta_o = c \Delta_t$, and $j_t = j_o$.
\footnote{The relative shift of the sub-grids e.g. in the $x$-$y$-plane can be established by a translation by one-half lattice spacing in $x$- or $y$-direction. In Eq. \eqref{planeonlattice} the translation in $x$-direction is chosen. This does not induce an asymmetry in the scheme. In each of the four equations the common $e$-factor for the center position in space and time can be canceled.  This neither changes the zeros of the characteristic determinant nor the spinor-component ratios of the eigensolutions.}
Insertion into the scheme above translates any discrete $x_o=ct$ derivative into a multiplication by $-\frac{2i}{\Delta_o}\sin(\frac{\omega \Delta_t}{2})$ and every discrete partial position $j=x,y,z$ derivative into a multiplication by a factor $\frac{2i}{\Delta_j} \sin(\frac{k_j \Delta_j}{2})$.  A time average leads to a factor $\cos(\frac{\omega \Delta_t}{2})$. 
\\
For fixed $\Delta_\mu$,  we therefore obtain the dispersion for the lattice from the continuum solution by the substitutions 
\begin{equation}
E=\hbar \omega \rightarrow E(\omega)=\frac{2\hbar }{\Delta_o}\sin(\frac{\omega \Delta_t}{2})~,\label{substE}
\end{equation}
\begin{equation}
p_j=\hbar k_j \rightarrow    P_j(k_j)= \frac{2 \hbar}{\Delta_j} \sin(\frac{k_j\Delta_j}{2})~,\label{substp}
\end{equation}
and 
\begin{equation}
m c^2 \rightarrow m(\omega)c^2 =m c^2 \cos(\frac{\omega \Delta_t}{2})~.\label{substm}
\end{equation}
\\
With these translation rules, the plane-wave solutions on the lattice can be obtained from the continuum solutions \cite{greiner, sakurai}.  In particular, the energy dispersion reads
\\
\begin{equation}
E(\omega)=\pm \sqrt{m(\omega)^2 c^4  + |c{\bf P}|^2}~,
\label{dispersion}
\end{equation}
were $\mathbf{P} = (P_1,P_2,P_3)$ is the momentum vector with the components defined in Eq. \ref{substp}.
This establishes the dispersion relation between $\omega$ and $k_j$.   Note that on the lattice the  $k_j$'s are restricted to the values $\big(-\frac{\pi}{\Delta_j}, +\frac{\pi}{\Delta_j}\big]$ and $\omega$ to $\big(-\frac{\pi}{\Delta_o}, +\frac{\pi}{\Delta_o}\big]$.
%??? okh
Equation \eqref{dispersion}  can be solved for $\omega$ by  taking the square and rearranging it to the final result for the dispersion relation
\begin{equation}
\hbar\omega=\pm \frac{2\hbar}{\Delta_t} \arcsin\big[X(k_j)\big]\quad \mbox{with} \quad X(k_j)=\sqrt{\frac{ (mc^2)^2 + |c {\bf P}|^2} {(mc^2)^2 + (\frac{2 \hbar  }{\Delta_t})^2}}~.
\label{pwdsip}
\end{equation}
\\
Moreover,
\begin{equation}
E(\omega)=\pm\frac{2\hbar}{\Delta_t} X(k_j)
\end{equation}
and 
\begin{equation}
m(\omega)c^2 =mc^2 \sqrt{1-X(k_j)^2}~.
\end{equation}
\\
\\
The spinor eigenket structure, as for the continuum model, is \cite{sakurai}

\begin{equation}
 {\vec \psi}(k_j)^{(1)} =  N\left(\begin{array}{c} 1 \\ 0\\ \frac{cP_z}{E(\omega) +m(\omega)c^2}\\ \frac{c(P_x+iP_y) }{E(\omega) +m(\omega)c^2}\end{array}\right) ~ ,\qquad  {\vec \psi}(k_j)^{(2)} =  N\left(\begin{array}{c} 0 \\ 1\\  \frac{c(P_x-iP_y) }{E(\omega) +m(\omega)c^2} \\ -\frac{cP_z}{E(\omega) +m(\omega)c^2} \end{array}\right) \nonumber ,
\end{equation}
\\
\begin{equation}
 {\vec \psi}(k_j)^{(3)} =  N\left(\begin{array}{c}  \frac{-cP_z}{\mid E(\omega)\mid +m(\omega)c^2}\\ -\frac{c(P_x+iP_y) }{\mid E(\omega)\mid +m(\omega)c^2} \\ 1 \\ 0\\\end{array}\right) ~ ,\qquad  {\vec \psi}(k_j)^{(4)} =  N\left(\begin{array}{c} -\frac{c(P_x-iP_y) }{\mid E(\omega)\mid +m(\omega)c^2} \\ \frac{cP_z}{\mid E(\omega)\mid +m(\omega)c^2}\\   0 \\ 1 \end{array}\right) ~ .
\end{equation}
\\
$N$ depends on the choice of normalization.

%???neu
The dispersion relation is shown in Fig. \ref{dispersion3D} for mass $m=0$ and two values of $k_z$.  The scheme displays a single cone centered around $k=0$ and preserves Kramers ($k_j\rightarrow -k_j$) and time-reversal symmetry.  Further symmetries depend on the ratios between the grid spacings $\Delta_\nu$.  
%???neu

\subsection{Stability}

A preliminary discussion of stability can be based on the dispersion relation associated with plane wave solutions (in absence of an external potential), Eq. \eqref{pwdsip}.  Instability occurs when $| X(k_j) |$ becomes larger than one, making  the two solutions for $\omega$ to become a pair of complex conjugate numbers.  Imposing
$$
\Bigg|  \frac{(mc^2)^2 + (2 \hbar c)^2 \sum_{j=x}^{z} \frac{1}{\Delta_j^2} }{ (mc^2)^2 + (2 \hbar c)^2\frac{1}{\Delta_o^2} }\Bigg| \leq 1~,
$$
one arrives at the condition 
$$
r_x^2+r_y^2+r_z^2\leq 1 ~,
$$
for $r_j=\Delta_o/\Delta_j$. Thus, for equal spatial differences $\Delta_x = \Delta_y = \Delta_z = \Delta$ this gives the condition 
$r=\Delta_o/\Delta \leq 1/\sqrt{3}$.

In order to discuss stability of this scheme in more generality, we seek a quantity (norm) which is conserved under the discrete time evolution.  
For this purpose we introduce several definitions and abbreviations to simplify notation.   Let us recall on which grid positions the spinor components are defined
\begin{align}\label{ABCDgrid}
A_j ~\mbox{with}~ j &~\in~ (j_x,j_y,j_z)~\cup~(j_x+1/2,j_y+1/2,j_z)~, \nonumber \\
B_j ~\mbox{with}~ j &~\in~ (j_x+1/2,j_y,j_z+1/2)~\cup~(j_x,j_y+1/2,j_z+1/2)~, \nonumber \\
C_j ~\mbox{with}~ j &~\in~ (j_x,j_y,j_z+1/2)~\cup~(j_x+1/2,j_y+1/2,j_z+1/2)~, \nonumber \\
D_j ~\mbox{with}~ j &~\in~ (j_x+1/2,j_y,j_z)~\cup~(j_x,j_y+1/2,j_z)~,
\end{align}
where $j_x$, $j_y$ and $j_z$ are integer numbers.
In what follows we consider a time step $\Delta_t$ where spinor  $A$ and $B$ are propagated from $j_o-1/2$ to $j_o+1/2$, 
and $C$ and $D$ are propagated from $j_o$ to $j_o+1$, using the short-hand notation $A_j^-, B_j^-,C_j^-,D_j^- \rightarrow A_j^+, B_j^+,C_j^+,D_j^+ $. 
Furthermore we define a scalar product between spinor components  $F^{j_o}_j$ and $G^{j_o'}_{j+j'}$  on the lattice in which a sum over all spatial sites $j$ is performed
\begin{equation}
(F^{j_o}, G^{j_o'})_{j'} =\sum_j F^{j_o}_j G^{j_o'*}_{j+j'}= \sum_j F^{j_o}_{j-j'} G^{j_o'*}_{j}. 
\end{equation}
This short-hand notation involves a summation over all the spatial sublattice sites $j$ on the time sheet $j_o$ which belong to spinor component $F$.  Also, $j'$ is limited to spatial shifts which connect the spatial sublattice of $F$ to the one associated with spinor component $G$.  Note that this scalar product depends on  $j_o$,  $j_o'$, the relative position vector $j'$, and the two spinor components involved. Here, however, we only need to consider products for up to nearest neighbor sites.  For simplicity, we write $(F^{j_o}, G^{j_o'})_{0}=(F^{j_o}, G^{j_o'})$.
\\ 
$j_i^\pm$ is defined as the vector shifting in the spatial direction $i$ by one half grid-spacing. Eg. $j_x^\pm := (j_x\pm 1/2,j_y,j_z)$. Finally we define for any spinor component
\begin{equation}
(\delta_i F)_j:= F_{j_i^+} -  F_{j_i^-}~. 
\end{equation}
\\
%AAAneu
Now we introduce the conserved functional:
\begin{lem}
The quadratic functional
\begin{align}
E_\mathbf{r} &:= (A,A) + (B,B) + (C,C) + (D,D) \nonumber \\ &~~- r_z \Re\Big[ (\delta_z C,A) -(\delta_z D,B) \Big]-r_x \Re\Big[(\delta_x D,A)  + (\delta_x C,B) \Big]\nonumber \\
&~~+  r_y\Re \Big[ i (\delta_y D,A) 
 - i (\delta_y C,B) \Big]
 \label{EPM}
\end{align}
is conserved under time propagation by the scheme \eqref{3+3Ddis}. Here we used the notation
$\mathbf{r}=(r_x,r_y,r_z)$.
\end{lem}
\begin{proof}\label{proofstability3D}
The proof follows the strategy from the (1+1)D paper \cite{hammer1D}.
%neuAAA
We begin by writing the scheme in short-hand notation using the definitions given above:
 \begin{align}
A^+_j-A^-_j +r_z(C^-_{j_z^+}-C^-_{j_z^-})+r_x(D^-_{j_x^+}-D^-_{j_x^-})- i r_y (D^-_{j_y^+}-D^-_{j_y^-})&=\frac{\Delta_o(mc^2+V)_j}{2i\hbar} (A^+_j+A^-_j )~,  \label{ss1} \\
B^+_j-B^-_j -r_z(D^-_{j_z^+}-D^-_{j_z^-})+r_x(C^-_{j_x^+}-C^-_{j_x^-})+i r_y (C^-_{j_y^+}-C^-_{j_y^-})&=\frac{\Delta_o(mc^2+V)_j}{2i\hbar} (B^+_j+B^-_j )~,   \label{ss2}\\
C^+_{j}-C^-_{j} +r_z(A^+_{{j}_z^+}-A^+_{{j}_z^-})+r_x(B^+_{{j}_x^+}-B^+_{{j}_x^-})- i r_y (B^+_{{j}_y^+}-B^+_{{j}_y^-})&=\frac{\Delta_o(-mc^2+V)^+_{j}}{2i\hbar} (C^+_{j}+C^-_{j} )~,   \label{ss3}\\
D^+_j-D^-_j -r_z(B^+_{j_z^+}-B^+_{j_z^-})+r_x(A^+_{j_x^+}-A^+_{j_x^-})+i r_y (A^+_{j_y^+}-A^+_{j_y^-})&=\frac{\Delta_o(-mc^2+V)^+_j}{2i\hbar} (D^+_j+D^-_j )~, \label{ss4}
\end{align}
where $V=V^{j_o}$, $V^+=V^{j_o+1/2}$, and analogously for $m$.
\\
Each of the  Eqs. \eqref{ss1} to  \eqref{ss4} is used to obtain an identity for the real part of a scalar product.  Eq. \eqref{ss1}  is multiplied by $(A^+_j+A^-_j )^*$ and summed over all lattice sites $j$.  Similarly Eq. \eqref{ss2}  is multiplied by $(B^+_j+B^-_j )^*$  and so on.  Adding up the real part of these four equations one obtains

\begin{align}
(A^+,A^+)-(A^-,A^-)+(B^+,B^+)-(B^-,B^-)+(C^+,C^+)-(C^-,C^-)+ (D^+,D^+)-(D^-,D^-)  \nonumber \\ 
+\Re \left\{ r_z\left[ (A^+ + A^-,\delta_z C^-)+(C^+ +C^-,  \delta_z A^+) -
 (B^++B^-,\delta_z D^-)-(D^+ + D^-,\delta_z B^+) \right] \right. \nonumber \\
\left.  + r_x\left[(A^+ + A^-,\delta_x D^-)+(D^+ + D^-,\delta_x A^+) +
 (B^+ + B^-,\delta_x C^-)+(C^+ + C^-,\delta_x B^+) \right] \right\}  \nonumber \\
-r_y \Im\left\{ (A^+ +A^-,\delta_y D^-)-(D^+ + D^-,\delta_y A^+) -
 (B^+ + B^-,\delta_y C^-)+(C^+ + C^-,\delta_y B^+)
 \right\}   =0~. 
 \label{E+-E-}
\end{align}
\\
Whenever possible, terms have been grouped in pairs of the form $\Re\{(F^++F^-,\delta_k G^-)+ (G^++G^-,\delta_k F^+)\}$ or $\Im\{(F^++F^-,\delta_k G^-)- (G^++G^-,\delta_k F^+)\}$.  
Each of the second terms in these pairs may be rewritten as follows (i.e.~by summation by parts)
$$ 
(G^++G^-,\delta_k F^+)= \sum_j (G^+_j+G^-_j)~(F^+_{j_k^+}-F^+_{j_k^-})^*= -(\delta_k G^+,F^+)-(\delta_k G^-,F^+) ~,
$$
We now observe that the contribution  arising from mixed terms in time $+$ and $-$  cancel when taking, respectively,  the real and imaginary part, and  Eq. \eqref{E+-E-} takes the form $E^+-E^-=0$ with $E$ given in Eq. \eqref{EPM}.
\end{proof}
%
%
%AAAneu
Next we utilize the conserved functional \eqref{EPM} to prove stability of the numerical scheme for arbitrary space- and time-dependent mass and potential.
\begin{prop}\label{Prop2}
Let $r_x+r_y+r_z< 1$ (e.g. using $r=r_x=r_y=r_z < 1/3$). 
Then, the spinor elements stay finite and are bounded for all time by the estimate
\begin{equation}\label{spinor-bound}
(A,A) + (B,B) + (C,C) + (D,D) \leq \frac{E^{0}_\mathbf{r}}{1- r_x - r_y - r_z}~.
\end{equation}
$E^{0}_\mathbf{r}$ is the squared norm \eqref{EPM} of the initial data $A^{-1/2}$, $B^{-1/2}$, $C^{0}$, and $D^{0}$.
\end{prop}
%neuAAA
\begin{proof}
We start estimating the functional $E_\mathbf{r}$:
\begin{align}
E_\mathbf{r} = E^{0}_\mathbf{r} &= (A,A) + (B,B) + (C,C) + (D,D) \nonumber \\ &~~- r_z \Re\Big[ (\delta_z C,A) -(\delta_z D,B) \Big]-r_x \Re\Big[(\delta_x D,A)  + (\delta_x C,B) \Big]\nonumber \\
&~~+ r_y\Re \Big[i (\delta_y D,A) 
 -i (\delta_y C,B) \Big]\nonumber \\
\nonumber \\
&\geq (A,A) + (B,B) + (C,C) + (D,D) \nonumber \\
&~~- r_z \Big|\Re\Big[ (\delta_z C,A)\Big]\Big| - r_z \Big|\Re\Big[(\delta_z D,B) \Big]\Big| - r_x \Big|\Re\Big[(\delta_x D,A)\Big]\Big|  - r_x \Big|\Re\Big[(\delta_x C,B) \Big]\Big|\nonumber \\
&~~- r_y \Big|\Re \Big[ (\delta_y D,A)\Big]\Big| -r_y \Big|\Re\Big[(\delta_y C,B) \Big]\Big|\nonumber \\
\nonumber \\
&\geq (A,A) + (B,B) + (C,C) + (D,D) \nonumber \\
&~~- r_z \Big|\Re\Big[ (C_{j_z^+},A)\Big]\Big| - r_z \Big|\Re\Big[(D_{j_z^+},B) \Big]\Big| - r_x \Big|\Re\Big[(D_{j_x^+},A)\Big]\Big|  - r_x \Big|\Re\Big[(C_{j_x^+},B) \Big]\Big|\nonumber \\
&~~- r_y \Big|\Re \Big[ (D_{j_y^+},A)\Big]\Big| -r_y \Big|\Re\Big[(C_{j_y^+},B) \Big]\Big|\nonumber \\
&~~- r_z \Big|\Re\Big[ (C_{j_z^-},A)\Big]\Big| - r_z \Big|\Re\Big[(D_{j_z^-},B) \Big]\Big| - r_x \Big|\Re\Big[(D_{j_x^-},A)\Big]\Big|  - r_x \Big|\Re\Big[(C_{j_x^-},B) \Big]\Big|\nonumber \\
&~~- r_y \Big|\Re \Big[ (D_{j_y^-},A)\Big]\Big| - r_y \Big|\Re\Big[(C_{j_y^-},B) \Big]\Big|\nonumber \\
\nonumber \\
&\geq (A,A) + (B,B) + (C,C) + (D,D) \nonumber \\
&~~- \big(r_x+r_y+r_z\big) \bigg((A,A) + (B,B) + (C,C) + (D,D)\bigg)~.
\end{align}
In the last line we used the inequality $2 |\Re\left\{(a,b)\right\}| \leq  \left\| a \right\|^2 + \left\|b\right\|^2$. And this yields the estimate \eqref{spinor-bound}.
\end{proof}
~\\
By comparison with the result from the ``reality condition" ($\Im\{\omega\} = 0$) of the dispersion (obtained for constant mass and potential), one observes that the condition in Proposition \ref{Prop2} for arbitrary space- and time-dependent coefficients gives a narrower bound for stability ({\it i.e.}  $r < 1/3$ vs. $r\leq 1/\sqrt{3}$). 
This is due to the coarse estimates in the short proof above. 
%AAAneu
And, in fact, we shall now improve the stability condition to $r\leq 1/\sqrt{3}$ for the case $r=r_x=r_y=r_z$:
\begin{prop}\label{Prop3}
Let $r=1/\sqrt{3}$. %$r=r_x=r_y=r_z=1/\sqrt{3}$. 
Then the scheme \eqref{3+3Ddis} for arbitrary space- and time-dependent mass and potential is stable. It satisfies the following estimate for all time:
$$
  \|(\tilde{A},\tilde{B})\|^2 + \|(\tilde{C},\tilde{D})\|^2 \le 2E^0,
$$
where $E^0$ is the ``energy'' of the initial data, as in \eqref{spinor-bound}. 
And on the l.h.s.~we use the following grid-averaged norm of the spinor components:
\begin{align}\label{CD-norm}
\|(\tilde{C},\tilde{D})\|^2:= &\sum_{j}
\left|\frac{C_{j+(0,0,\frac{1}{2})}+D_{j+(\frac{1}{2},0,0)}- i
D_{j+(0,\frac{1}{2},0)}}{2\sqrt{3}}+\frac{C_{j-(0,0,\frac{1}{2})}+D_{j-(\frac{1}
{ 2},0,0)}- i D_{j-(0,\frac{1}{2},0)}}{2\sqrt{3}}\right|^2\\\nonumber
&\!\!\!\!\!+\sum_{j}
\left|\frac{D_{j+(\frac{1}{2},0,1)}-C_{j+(1,0,\frac{1}{2})}- i
C_{j+(\frac{1}{2},\frac{1}{2},\frac{1}{2})}}{2\sqrt{3}}+\frac{D_{j+(\frac{1}{2},
0,0)}-C_{j+(0,0,\frac{1}{2})}- i
C_{j+(\frac{1}{2},-\frac{1}{2},\frac{1}{2})}}{2\sqrt{3}}\right|^2~,\\\nonumber
\end{align}
with $j\in \mathbb{Z}^3 \cup (\mathbb{Z}+1/2)^2\times\mathbb{Z}$ in the sum.
$\|(\tilde{A},\tilde{B})\|^2$ is defined analogously, subject to an index shift detailed in \eqref{ABCDgrid}.
\end{prop}
This stability proof is more lengthy and hence deferred to \ref{A}.
%neuAAA
~\\

%%%%%%%%%%%%%%%%%%%%%%%%%%%%%%%%%%%%%%%%%%%%%%%%%%%%%%%%%%%%%%%%%%%%%%%%%%%%%%%%%%%%
 
\section{Numerical Scheme for the (2+1)D Dirac Equation}\label{sect3}

In some cases it is not necessary to solve the full (3+1)D Dirac equation.  For example, under translational invariance in one spatial direction  (when mass and potential terms are independent of this space coordinate)  one can perform a partial Fourier transform within the scheme above, essentially leading to the substitution Eq. \eqref{substp}.
The spin degree of freedom may be unimportant or it may be locked to the orbital motion, as is the case for the effective model for topological insulator (TI) surface states \cite{qi}. Then,  a two component spinor in a  (2+1)D  model is adequate to capture the underlying physics.  In some cases, low-dimensional models have been used simply because they are easier to handle.  
%??? okh

Once more, we use the Schr\"{o}dinger form of the (2+1)D Dirac equation
\begin{equation}
i \hbar c \partial_t \mbox{\boldmath$\psi$}(x,y,t) = \bigg[- i \hbar c \sigma_x \frac{\partial}{\partial x} - i \hbar c  \sigma_y \frac{\partial}{\partial y}  + \sigma_z m(x,y,t)+\mathbb{1}_2 V(x,y,t)\bigg] \mbox{\boldmath$\psi$}(x,y,t)~,
\end{equation}
where $\mbox{\boldmath$\psi$}(x,y,t)\equiv(u,v)\in\mathbb{C}^2$ is a $2$-component spinor. For topological insulators, the two components are associated with spin=$1/2$, whereby the physical spin quantization axis is $\mathbf{S}\propto\mathbf{\hat z}\times \mbox{\boldmath$\sigma$}$. Here, $\mathbf{\hat z}$ is the unit-normal vector to the surface and  $c$ constitutes the effective velocity of the quasi-particles in the effective model. $V\in\mathbb{R}$ is the scalar potential.
\begin{figure}[t!]
\centering
%
%????new
\includegraphics[width=8cm]{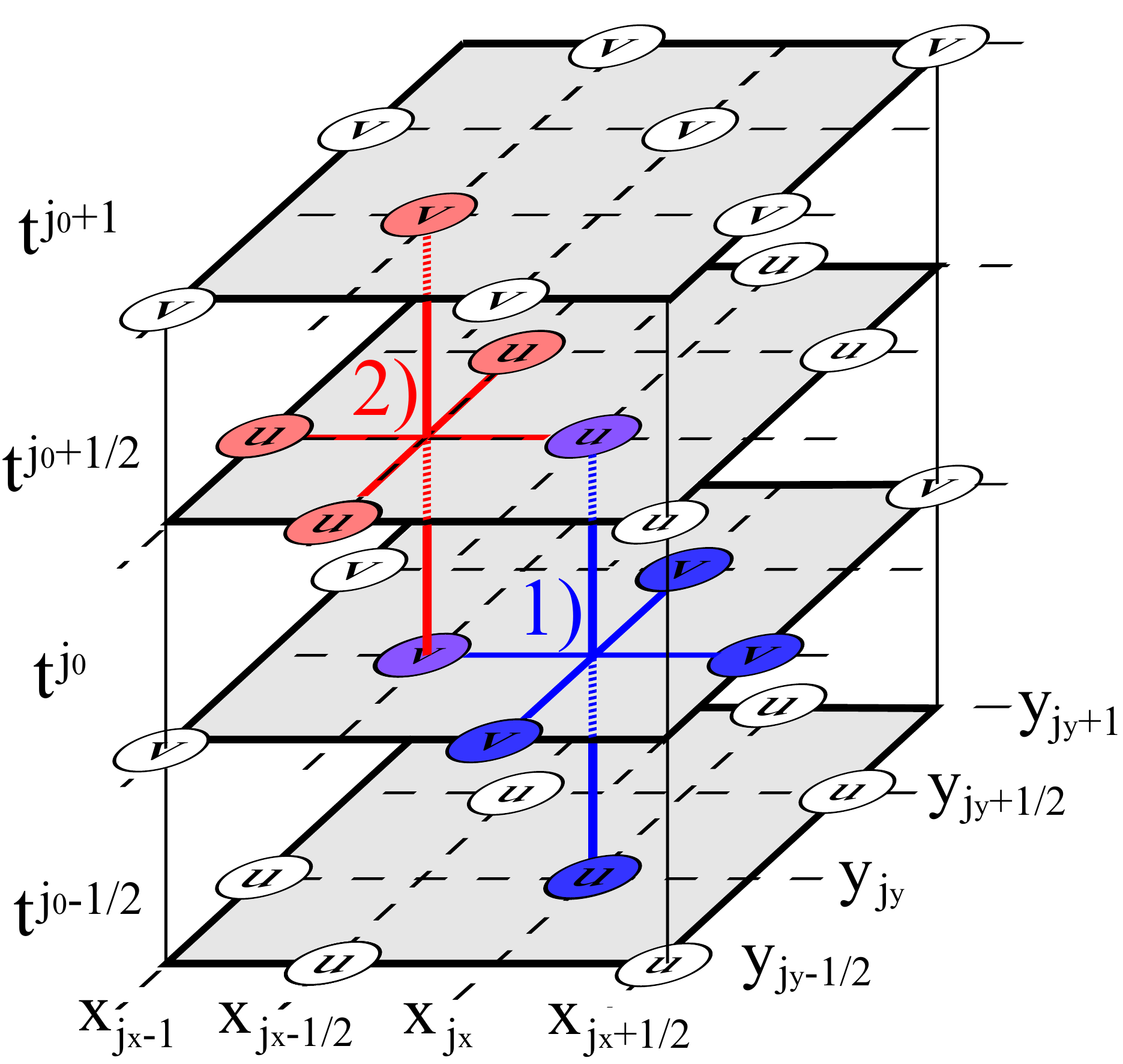}
\caption{(color online). Leap-frog time stepping on a time and space staggered grid for the (2+1)D two component spinor Dirac equation: Time propagation is shown to take place in vertical direction.  \textcolor{blue}{1)} In the first sub-step new $u$-components \textcolor{blue}{(blue/dark gray)} on time sheet $t^{j_o+1/2}$ are computed from $u$ on time sheet $t^{j_o-1/2}$ and the centered spatial differences of the $v$ components living on time sheet $t^{j_o}$. \textcolor{red}{2)} In the second sub-step \textcolor{red}{(red/light gray)} the centered spatial differences of $u$ on time sheet  $t^{j_o+1/2}$ together with the old $v$-components on time sheet $t^{j_o}$ are used to update for the new $v$'s on time sheet $t^{j_o+1}$.}\label{scheme2D}
\end{figure}
%
%????new

Earlier on and mainly for the treatment of rectangular structures we have presented a scheme which shows perfect dispersion along the coordinate axes for the free mass-less case, but with a second Dirac cone equally shared by the four corners of the first Brillouin zone \cite{hammer2D2cone,hammerAPL}.  It is therefore ideally suited for TI surface states with rectangular structuring.  
Here a different staggering of the grid in space is applied to the $u$- and $v$-component of the spinor $\mbox{\boldmath$\psi$} = (u,v)$ which eliminates the second cone.  It is represented graphically in Fig. \ref{scheme2D} and 
follows the procedure for the (3+1)D case:  A progression in time by $\Delta_o=c\Delta_t$ involves two initial, $t=j_o-1/2, j_o$, and two final time sheets, $t=j_o+1/2, j_o+1$.  Spinor component $u$ lives on a face-centered  rectangular (fcr) lattice (fcr-u) on the half-integer time sheets. Spinor component $v$ lives on the fcr lattice (fcr-v) formed by the set of midpoint positions of lines connecting nearest neighbors of the simple rectangular sub-lattice of $u$, however, shifted vertically onto the integer time sheets.  In other words, fcr-u is converted into fcr-v by a shift by $(i+1/2)\Delta_x {\hat e}_x $ (or $ (i+1/2)\Delta_y{\hat e}_y$)  followed by $\pm (j+1/2) \Delta_o {\hat e}_o, i, j \in\mathbb{Z}$, and vice versa.

The time progression by $\Delta_o$ is executed in two steps.  First $u$ is propagated from time sheet $j_o-1/2$ to $j_o+1/2$ (both supporting fcr-u) by forming symmetric 
$x$- and $y$-derivatives of $v$ using fcr-v on time sheet $j_o$.  In the second step, $v$ on fcr-v on time sheet $j_o$ is propagated to fcr-v on time sheet $j_o+1$ using symmetric $x$- and $y$-derivatives of $u$ living on fcr-u on time sheet $j_o+1/2$.  
According to the discussion above, one has to propagate
$u_{j_x,j_y}^{j_o-1/2}$, $u_{j_x+1/2,j_y+1/2}^{j_o-1/2}$, $v_{j_x+1/2,j_y}^{j_o}$, and $v_{j_x,j_y+1/2}^{j_o}$ for $j_\nu\in \mathbb{Z}$.  
Using second order accurate symmetric approximations for the derivatives one arrives at the following scheme for this grid
\begin{align}\label{scheme2Deq}
\frac{u_{j_x,j_y}^{j_o+1/2}-u_{j_x,j_y}^{j_o-1/2}}{\Delta_t}=&\bigg(\frac{m+V}{i \hbar c}\bigg)_{j_x,j_y}^{j_o}\frac{u_{j_x,j_y}^{j_o+1/2}+u_{j_x,j_y}^{j_o-1/2}}{2}\\
-&\frac{(v_{j_x+1/2,j_y}^{j_o}-v_{j_x-1/2,j_y}^{j_o})}{\Delta_x}+i\frac{(v_{j_x,j_y+1/2}^{j_o}-v_{j_x,j_y-1/2}^{j_o})}{\Delta_y}~,\nonumber
\end{align}
\begin{align}\label{scheme2ndline}
\frac{v_{j_x-1/2,j_y}^{j_o+1}-v_{j_x-1/2,j_y}^{j_o}}{\Delta_t}=-&\bigg(\frac{m-V}{i \hbar c}\bigg)_{j_x-1/2,j_y}^{j_o+1/2}\frac{v_{j_x-1/2,j_y}^{j_o+1}+v_{j_x-1/2,j_y}^{j_o}}{2}\\
-&\frac{(u_{j_x,j_y}^{j_o+1/2}-u_{j_x-1,j_y}^{j_o+1/2})}{\Delta_x} -i\frac{(u_{j_x-1/2,j_y+1/2}^{j_o+1/2}-u_{j_x-1/2,j_y-1/2}^{j_o+1/2})}{\Delta_y}~,\nonumber
\end{align}
and  analogous equations for the other sub-grid
\begin{align}\label{scheme3rdline}
\frac{u_{j_x-1/2,j_y-1/2}^{j_o+1/2}-u_{j_x-1/2,j_y-1/2}^{j_o-1/2}}{\Delta_t}=&\bigg(\frac{m+V}{i \hbar c}\bigg)_{j_x-1/2,j_y-1/2}^{j_o}\frac{u_{j_x-1/2,j_y-1/2}^{j_o+1/2}+u_{j_x-1/2,j_y-1/2}^{j_o-1/2}}{2}\\
-&\frac{(v_{j_x,j_y-1/2}^{j_o}-v_{j_x-1,j_y-1/2}^{j_o})}{\Delta_x} +i\frac{(v_{j_x-1/2,j_y}^{j_o}-v_{j_x-1/2,j_y-1}^{j_o})}{\Delta_y}~,\nonumber
\end{align}
\begin{align}\label{scheme4thline}
\frac{v_{j_x,j_y-1/2}^{j_o+1}-v_{j_x,j_y-1/2}^{j_o}}{\Delta_t}=-&\bigg(\frac{m-V}{i \hbar c}\bigg)_{j_x,j_y-1/2}^{j_o+1/2}\frac{v_{j_x,j_y-1/2}^{j_o+1}+v_{j_x,j_y-1/2}^{j_o}}{2}\\ &-\frac{(u_{j_x+1/2,j_y-1/2}^{j_o+1/2}-u_{j_x-1/2,j_y-1/2}^{j_o+1/2})}{\Delta_x} -i\frac{(u_{j_x,j_y}^{j_o+1/2}-u_{j_x,j_y-1}^{j_o+1/2})}{\Delta_y}~.\nonumber
\end{align}
Note that the $u$-component defined for the discrete time indices $j_o-1/2\in \mathbb{Z}$ `lives' on the discrete space gridpoints $(j_x,j_y)\in\mathbb{Z}^2$ and on $(j_x-1/2,j_y-1/2)\in\mathbb{Z}^2$,  whereas the $v$-component defined for $j_o\in \mathbb{Z}$ is defined for space indices $(j_x-1/2,j_y)\in\mathbb{Z}^2$ and $(j_x,j_y-1/2)\in\mathbb{Z}^2$.  
\\
For constant coefficients $m$ and $V$ von Neumann stability analysis reveals the stability of the finite difference scheme \cite{strikwerda}. Moreover, the growth factor shows that an imaginary potential can be utilized to model an absorbing boundary layer.
%??? okh
 Fourier transformation in the spatial coordinates gives
\begin{align}
\underbrace{\left(\begin{array}{cc}
\frac{1}{\Delta_t}-\frac{m+V}{i \hbar c} & 0\\
\frac{2 i}{\Delta_x}\sin{\frac{k_x \Delta_x}{2}}+i \frac{2 i}{\Delta_y}\sin{\frac{k_y \Delta_y}{2}} & \frac{1}{\Delta_t}+\frac{m-V}{i \hbar c} 
\end{array}\right)}_{=:S}
\left(\begin{array}{c}
\tilde{u}^{+}\\
\tilde{v}^{+}
\end{array}\right)\\\nonumber
+\underbrace{\left(\begin{array}{cc}
-\frac{1}{\Delta_t}-\frac{m+V}{i \hbar c} & \frac{2 i}{\Delta_x}\sin{\frac{k_x \Delta_x}{2}}-i \frac{2 i}{\Delta_y}\sin{\frac{k_y \Delta_y}{2}} \\
0 & -\frac{1}{\Delta_t}+\frac{m-V}{i \hbar c}
\end{array}\right)}_{=:T}
\left(\begin{array}{c}
\tilde{u}^{-}\\
\tilde{v}^{-}
\end{array}\right)=0~\label{stabilityeq}.
\end{align}
%
%??? okh
One defines the amplification matrix $G = -S^{-1} T = S^{-1} S^* ~$.  Its eigenvalues are the growth factors (written for $\hbar=1, c= 1$ and $\Delta_x = \Delta_y = \Delta$)
\be
\lambda_{\pm} = P/2 \pm \sqrt{(P/2)^2-Q}~,
\ee
with
\be
P = \mbox{tr}[G] =-2 \Big[(m^2-V^2)\Delta_t^2 - 4 c^2 (1-\Delta_t^2) - 4 c^2 \Delta_t^2 (\Delta_y^2 \cos k_x \Delta_x + \Delta_x^2 \cos k_y \Delta_y)/(\Delta_x \Delta_y)^2 \Big]/N~,
\ee
\be
Q  = \mbox{det}[G] = (4 c^2 + (m^2-V^2)\Delta_t^2 - 4 i c V \Delta_t)/N
\ee
and
\be
N= 4 c^2 + (m^2-V^2)\Delta_t^2 + 4 i c V \Delta_t~.
\ee\\
For $\Delta_t/\Delta \leq 1/\sqrt{2}$ and $m, V \in\R$ the absolute value of the growth factors is $1$. Whereas for $V\in\mathbb{C}$ it depends on the sign of the imaginary part whether their absolute value is greater or lesser than $1$. The latter can be used for \emph{absorbing boundary layers}. The size of the absolute value of the growth (or in this case damping) factors is exemplified in Fig. \ref{abc}.
%???% okh
\begin{figure}[t!]
\centering
\includegraphics[width=13cm]{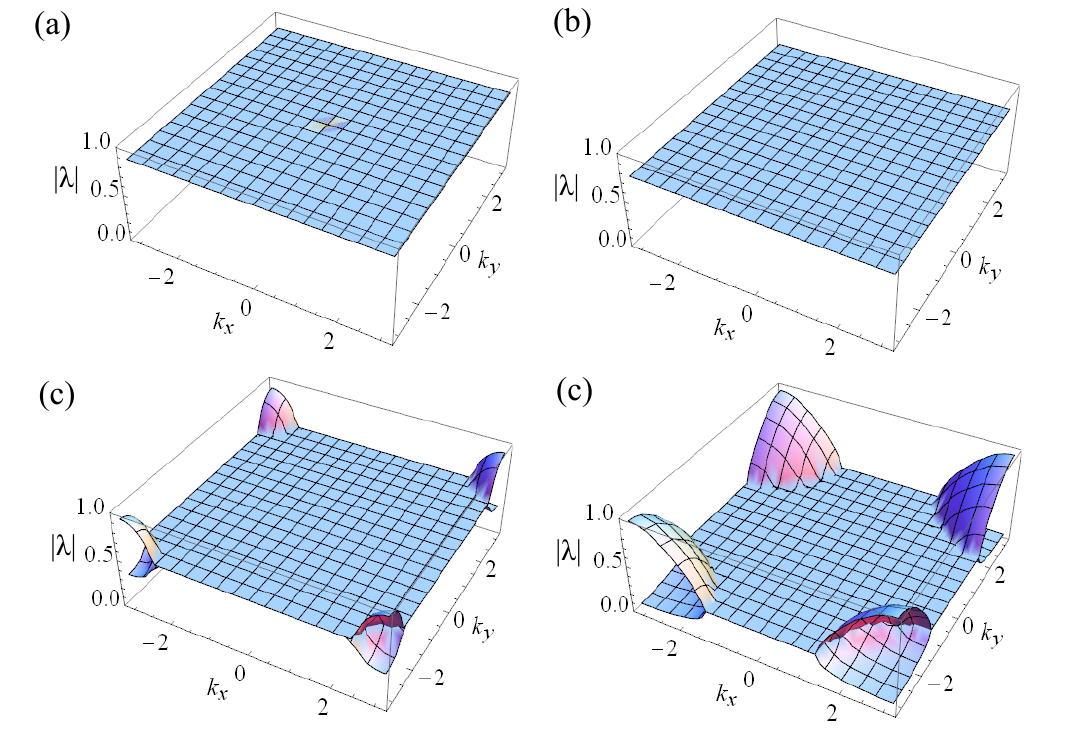}
\caption{(color online). Absolute value of the eigenvalues $|\lambda_\pm|$ of the growth matrix $G$  for $m=0$ and $\Delta_t/\Delta= 0.99/\sqrt{2}$. (a) $V=0.1 i$, (b) $V=0.2 i$, (c) $V=0.4 i$ and (d) $V=0.8 i$. }\label{abc}
\label{lambda-imag-V}
\end{figure}
\\
\noindent  The free-particle \emph{dispersion relation} is  revealed directly by using a plane wave ansatz\\ $u^{j_o+1}_{j_x+1,j_y+1} = e^{i (\omega \Delta_t - k_x \Delta_x - k_y \Delta_y)}u^{j_o}_{j_x,j_y}$ (and similarly for $v$) leading to the substitutions Eqs. \eqref{substE}-\eqref{substm} and the dispersion relation\\
\begin{equation}
\hbar \omega = \pm \frac{2 \hbar}{\Delta_t} \arcsin \Bigg\{ \sqrt{ \frac{1}{(m c^2)^2 + 2 \hbar c /\Delta_t} \bigg[(m c^2)^2 + \Big(\frac{2 \hbar}{\Delta_x}\sin \frac{k_{x} \Delta_x}{2}\Big)^2 + \Big(\frac{2 \hbar}{\Delta_y}\sin \frac{k_{y} \Delta_y}{2}\Big)^2\bigg]}\Bigg\}~.
\end{equation}\\
\\
The naive discrete expression for the norm $\left\|\psi\right\|:= \sqrt{\sum_{j_x,j_y} |u_{j_x,j_y}|^2
+ |v_{j_x,j_y}|^2}$ in general oscillates around its mean value and is conserved in time only on average.  We use a multiplication technique to identify an exactly  \emph{conserved functional} and prove stability for arbitrary space and time dependent coefficients,  as for the (3+1)D case above.
 Let us again introduce a short-hand notation using 
\begin{align}
u_j ~\mbox{with}~ j &~\in~ (j_x,j_y)~\cup~(j_x+1/2,j_y+1/2)~, \nonumber \\
v_j ~\mbox{with}~ j &~\in~ (j_x+1/2,j_y)~\cup~(j_x,j_y+1/2)~. \nonumber \\
\end{align}
We define spatial difference operators as $(\delta_x f^{j_o})_{j_x,j_y} = f_{j_x+1/2,j_y}^{j_o}-f_{j_x-1/2,j_y}^{j_o}$, $(\delta_y f^{j_o})_{j_x,j_y} = f_{j_x,j_y+1/2}^{j_o}-f_{j_x,j_y-1/2}^{j_o}$ and $(\delta_\pm f^{j_o})_{j_x,j_y} = (\delta_x f^{j_o})_{j_x,j_y}  \pm i(\delta_y f^{j_o})_{j_x,j_y} $. We define the inner product $(u^{j_o},v^{j_o'})_{j'} := \sum_{j} u_{j}^{j_o} v^{*j_o'}_{j+j'}=\sum_{j} u_{j-j'}^{j_o}  v^{*j_o' }_{j}$ on $l^2(\mathbb{Z}^2;\mathbb{C})$ and the notation $\left\|u^{j_o} \right\|^2:=(u^{j_o} ,u^{j_o} )$, with the sum over $j$ running  over all spatial lattice points on the  time sheet $j_o$.  $j'$ again denotes a displacement  vector connecting the two spatial sublattices of $u$ and $v$.    

%AAAneu 
In analogy to \eqref{EPM} we now define a conserved functional for the scheme \eqref{scheme2Deq}-\eqref{scheme4thline}:
\begin{lem}
Let $r=r_x=r_y$. Then the functional
\begin{equation}
E_r^{j_o} := \left\|u^{j_o+1/2}\right\|^2 + \left\|v^{j_o+1}\right\|^2 + r \Re \Big[(\delta_- u^{j_o+1/2},v^{j_o+1})\Big] = \mbox{const}= E_r^0~\label{E}
\end{equation}
is conserved under time propagation.
\end{lem}
\begin{proof}
The real part of the inner product of \eqref{scheme2Deq} with $({u}^{j_o+1/2}+{u}^{j_o-1/2})$  gives
\begin{equation}
\left\|u^{j_o+1/2}\right\|^2-\left\|u^{j_o-1/2}\right\|^2 + r \Re\Big[(\delta_+ v^{j_o},u^{j_o+1/2}+u^{j_o-1/2})\Big]  = 0~.\label{firsteq}
\end{equation}
Analogously, \eqref{scheme2ndline} is multiplied by $(v^{j_o+1}+v^{j_o})^*$ and again the real part is taken to give 
\begin{equation}
\left\|v^{j_o+1}\right\|^2-\left\|v^{j_o}\right\|^2 + r \Re\Big[(\delta_- u^{j_o+1/2},v^{j_o+1}+v^{j_o})\Big]  = 0~.\label{secondeq}
\end{equation}
Performing a summation by parts with vanishing "boundary terms" at infinity gives
\begin{equation}
 \Re \Big[(\delta_x v^{j_o}, u^{j_o+1/2}+u^{j_o-1/2})\Big]= -\Re\Big[(\delta_x u^{j_o+1/2}+\delta_x u^{j_o-1/2}, v^{j_o})\Big]~,
\end{equation}
\begin{equation}
 \Re \Big[(i \delta_y v^{j_o}, u^{j_o+1/2}+u^{j_o-1/2})\Big]=\Re \Big[i~(u^{j_o+1/2}+u^{j_o-1/2}, \delta_y v^{j_o})^*\Big]=\Re\Big[(i \delta_y u^{j_o+1/2}+i \delta_y u^{j_o-1/2}, v^{j_o})\Big]~.
\end{equation}
Finally,  adding Eq. \eqref{firsteq} and Eq. \eqref{secondeq} leads to
\begin{equation}
\left\|u^{j_o+1/2}\right\|^2 + \left\|v^{j_o+1}\right\|^2 + r \Re \Big[(\delta_- u^{j_o+1/2},v^{j_o+1})\Big] = \left\|u^{j_o-1/2}\right\|^2 + \left\|v^{j_o}\right\|^2 + r \Re \Big[(\delta_- u^{j_o-1/2},v^{j_o})\Big]~.
\end{equation}
\end{proof}
%neuAAA

\begin{figure}[t!]
\centering
\includegraphics[width=13cm]{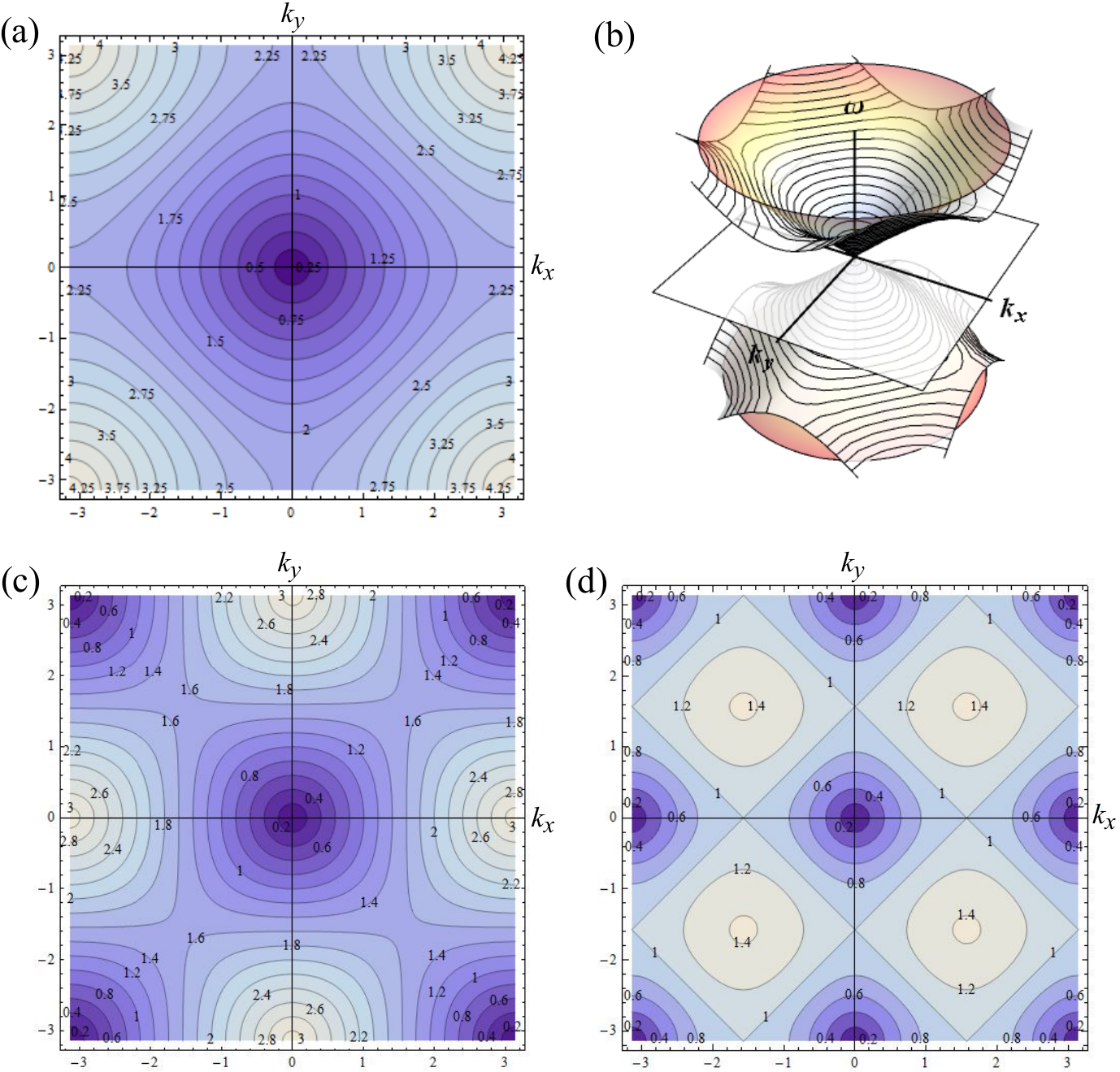}
\caption{(color online). Dispersion relation of the leap-frog staggered-grid finite difference scheme for the (2+1)D Dirac equation and comparison to other schemes from the literature. $\Delta_x = \Delta_y = \hbar = 1$, $\Delta_o = 1/\sqrt{2}$, $m=0$. (a) Contour plot of the positive energy dispersion relation. (b) Comparison of the numerical dispersion relation with the exact Dirac cone dispersion. (c) Dispersion relation of the leap-frog staggered-grid finite difference scheme with two Dirac cones \cite{hammer2D2cone}. (d) Dispersion relation of a centered differences in space and Crank-Nicolson in time scheme, without grid staggering, showing four Dirac cones.}
\label{dispersion2D}
\end{figure}
%
%
%AAAneu
Similar to the (3+1)D case in Section \ref{proofstability3D}, we can use this result to prove stability for an arbitrary space- and time-dependent mass and potential.
\begin{prop}\label{Prop4}
Let $r_x+r_y< 1$ (e.g.~using $r_x=r_y <  1/2$). Then, the (2+1)D--scheme is stable and satisfies the estimate
%\begin{equation}
$$
\left\|u^{j_o+1/2}\right\|^2 + \left\|v^{j_o+1}\right\|^2 \leq \frac{E^0_\mathbf{r}}{1- r_x - r_y}
$$
%\end{equation}
for all time.
\end{prop}
\begin{proof}
We start estimating the conserved functional
\begin{align}
E^0_\mathbf{r}=E^{j_o}_\mathbf{r}:=\left\|u^{j_o+1/2}\right\|^2 + \left\|v^{j_o+1}\right\|^2 &+ r_x \Re \Big[(\delta_x u^{j_o+1/2},v^{j_o+1})\Big] -  r_y \Re \Big[(i \delta_y u^{j_o+1/2},v^{j_o+1})\Big]\nonumber \\
\nonumber \\
\geq\left\|u^{j_o+1/2}\right\|^2 + \left\|v^{j_o+1}\right\|^2 &- r_x~ \Big|\Re \Big[(\delta_x u^{j_o+1/2},v^{j_o+1})\Big]\Big| - r_y~ \Big|\Re \Big[(i\delta_y u^{j_o+1/2},v^{j_o+1})\Big]\Big|\nonumber \\
\nonumber \\\geq \left\|u^{j_o+1/2}\right\|^2 + \left\|v^{j_o+1}\right\|^2 &- \big(r_x+r_y\big)\Big(\left\|u^{j_o+1/2}\right\|^2 + \left\|v^{j_o+1}\right\|^2 \Big)~\nonumber, 
\end{align}
where we have used the inequality $2 |\Re\left\{(a,b)\right\}| \leq  \left\| a \right\|^2 + \left\|b\right\|^2$. And this yields the claimed result.
\end{proof}
Again, the condition in Proposition \ref{Prop4} is too restrictive for constant mass and potential.  From the reality condition for the free-particle dispersion, one has $r_x^2+r_y^2 \leq 1$.  For the special case $r=r_x=r_y$ this  less restrictive stability condition reads $r \leq 1/\sqrt{2}$. As in Proposition \ref{Prop3} we shall now prove that this stability condition actually also holds for an arbitrary space- and time-dependent mass and potential. 

To this end we first define an averaged norm for the spinor components on the grid, pertaining to the special value $r=1/\sqrt{2}$:
\be
\left\|\tilde{u}\right\|^2 := \sum_{j_x,j_y} \bigg|\frac{u_{j_x+1,j_y}-i u_{j_x+1/2,j_y+1/2}}{2\sqrt{2}} + \frac{u_{j_x,j_y}-i u_{j_x+1/2,j_y-1/2}}{2\sqrt{2}}\bigg|^2~.
\ee
Here and in the following equations the summation runs over $j=(j_x,j_y) \in \mathbb{Z}^2 \cup (\mathbb{Z}+1/2)^2$. 

\begin{lem}
$\left\|\tilde{u}\right\|^2$ is a norm on $l^2\big[\mathbb{Z}^2 \cup (\mathbb{Z}+1/2)^2\big]$.
\end{lem}
\begin{proof}
Assume $\left\|\tilde{u}\right\|^2=0$, which leads to
\be
u_{j_x+1,j_y}-i u_{j_x+1/2,j_y+1/2} + u_{j_x,j_y}-i u_{j_x+1/2,j_y-1/2} = 0~.\label{unormeq0}
\ee
This is a leap-frog scheme $u_{j_x+1/2,j_y-1/2},u_{j_x,j_y} \rightarrow u_{j_x+1,j_y},u_{j_x+1/2,j_y+1/2}$.
To show that is has no non-trivial solution in $l^2$ we observe that the real part of $u\big|_{\mathbb{Z}^2}$ is only coupled to the imaginary part of $u\big|_{(\mathbb{Z}+1/2)^2}$:
\be
\Re u_{j_x,j_y} + \Im u_{j_x+1/2,j_y-1/2} = -(\Re u_{j_x+1,j_y} + \Im u_{j_x+1/2,j_y+1/2}) = \Re u_{j_x+1,j_y+1} + \Im u_{j_x+3/2,j_y+1/2} = \cdots\label{reim}
\ee
The $l^2$-summabilty of $u$ implies Eq. \eqref{reim}$~=0$ giving
\be
\Re u_{j_x,j_y} =- \Im u_{j_x+1/2,j_y-1/2}~.\label{unormeq1}
\ee
Now we shift \eqref{unormeq0} by $(\frac{1}{2},-\frac{1}{2})$ and take the
imaginary part:
\be
\Im u_{j_x+1/2,j_y-1/2} -  \Re u_{j_x+1,j_y-1} = -(\Im u_{j_x+3/2,j_y-1/2} -  \Re u_{j_x+1,j_y}) = \Im u_{j_x+3/2,j_y+1/2} -  \Re u_{j_x+2,j_y}=  \cdots=0~, \label{imre}
\ee
again due to the $l^2$-summabilty. Combining \eqref{unormeq1} and \eqref{imre}
with their integer grid-shifts yields
\be
\Re u_{j_x,j_y} = - \Im u_{j_x+1/2,j_y-1/2} = \Re u_{j_x+1,j_y-1} = - \Im u_{j_x+3/2,j_y-3/2}=\cdots~.\label{sum2}
\ee
By the $l^2$-summabilty, Eq. \eqref{sum2}$=0$.\\
\\
Analogously, the imaginary part of $u\big|_{\mathbb{Z}^2}$ is only coupled to
the real part of $u\big|_{(\mathbb{Z}+1/2)^2}$. The same arguments as before
lead to
\be
\Im u_{j_x,j_y} =  \Re u_{j_x+1/2,j_y-1/2} = \Im u_{j_x+1,j_y-1} =  \Re
u_{j_x+3/2,j_y-3/2} = \cdots=0~.
\ee
Hence $\tilde{u}=0$.
\end{proof}

Now we can formulate our stability result:
\begin{prop}
Let $r=r_x=r_y=1/\sqrt{2}$ hold in \eqref{scheme2Deq}-\eqref{scheme4thline}. Then this (2+1)D--scheme is stable and satisfies for all time
$$
  \left\|\tilde{u}\right\|^2 + \left\|\tilde{v}\right\|^2 \le 2E^0~.
$$
Here, $E^0$ is the ''energy''
\begin{align}
E := \sum_{j_x,j_y}  |u_{j_x,j_y}|^2 + |v_{j_x+1/2,j_y}|^2 + \frac{1}{\sqrt{2}} \Re \bigg[ \big(u_{j_x+1,j_y}-u_{j_x,j_y}-i u_{j_x+1/2,j_y+1/2}+i u_{j_x+1/2,j_y-1/2}, v_{j_x+1/2,j_y}\big)\bigg]
\end{align} 
of the initial data $u^{-1/2},\,v^0$ (use $r=1/\sqrt{2}$ in Eq. \eqref{E}).
\end{prop}
\begin{proof}
We use the same strategy as in the (1+1)D version of the scheme \cite{hammer1D}.
We first define the auxiliary quantity
\begin{align}
\tilde{E} :&= \frac{1}{2} \sum_{j_x,j_y} \bigg|\frac{u_{j_x+1,j_y}-i u_{j_x+1/2,j_y+1/2}}{\sqrt{2}}+v_{j_x+1/2,j_y}\bigg|^2 
+ \frac{1}{2} \sum_{j_x,j_y} \bigg|\frac{u_{j_x,j_y}-i u_{j_x+1/2,j_y-1/2}}{\sqrt{2}}-v_{j_x+1/2,j_y}\bigg|^2\\\nonumber
&= \sum_{j}  |u_{j}|^2 + |v_{j}|^2\\\nonumber
&~~+\frac{1}{2\sqrt{2}} \sum_{j_x,j_y} \bigg[~(u_{j_x+1,j_y}-i u_{j_x+1/2,j_y+1/2})~ v_{j_x+1/2,j_y}^* + (u_{j_x+1,j_y}-i u_{j_x+1/2,j_y+1/2})^* v_{j_x+1/2,j_y}\\\nonumber
&\qquad\qquad\qquad-(u_{j_x,j_y}-i u_{j_x+1/2,j_y-1/2})~ v_{j_x+1/2,j_y}^* - (u_{j_x,j_y}-i u_{j_x+1/2,j_y-1/2})^* v_{j_x+1/2,j_y}\bigg]\\\nonumber
&~~+\frac{i}{4} \sum_{j_x,j_y} \bigg[u_{j_x+1,j_y}u_{j_x+1/2,j_y+1/2}^*
- u_{j_x+1,j_y}^* u_{j_x+1/2,j_y+1/2}\\\nonumber
&\qquad\qquad\quad+ u_{j_x,j_y} u_{j_x+1/2,j_y-1/2}^*
- u_{j_x,j_y}^* u_{j_x+1/2,j_y-1/2}\bigg]~,
\end{align}
where the last summation equals zero because the terms cancel, e.g. the first term with $j_x=j_y=0$ cancels with the last one having $j_x=j_y=1/2$. 
Thus, one gets $\tilde{E}=E$.\\
\\
Using $\frac{1}{4}\left\|a_1 + a_2\right\|^2 \leq \frac{1}{4} \big(\left\|a_1 + b\right\|+\left\|a_2 - b\right\|\big)^2\leq \frac{1}{2} \left\|a_1 + b\right\|^2 + \frac{1}{2} \left\|a_2 - b\right\|^2~$ gives
$\left\|\tilde{u}\right\|^2 \leq \tilde{E} = E^0.$
The symmetry of the scheme and of $E$ in $u$ and $v$ yields also: $\left\|\tilde{v}\right\|^2\leq \tilde{E} = E^0$.
\end{proof}
%
%neuAAA

%%%%%%%%%%%%%%%%%%%%%%%%%%%%%%%%%%%%%%%%%%%%%%%%%%%%%%%%%%%%%%%%%%%%%%%%%%%%%%%%
%
\section{Gauge-Invariant Introduction of a Vector Potential}\label{sect4}

In the previous exposition, only the presence of a scalar potential has been considered.  The treatment of a Dirac fermion under the influence of a general electromagnetic field, however,  calls for the use of the electromagnetic four-vector potential, {\it i.e.}  both scalar and vector potential.  
The incorporation of the vector potential ${\bf {\cal A}}(x,y,z,t)$ poses no problem for the two schemes presented above.  It can be accomplished in the same fashion as discussed recently in conjunction with our (1+1)D and (2+1)D schemes, following the standard Peierls substitution for the spinor components $\psi$ on the lattice \cite{peierls,graf,hammer1D,hammer2D2cone},
%??? okh
\begin{eqnarray}
\psi^{t_j,z_j}_{x_j,y_j} \rightarrow  {\hat \psi}^{t_j,z_j}_{x_j,y_j} := \psi^{t_j,z_j}_{x_j,y_j} \exp\{ -i  a^{t_j,z_j}_{x_j,y_j}\} ~.
\label{peierlss}
\end{eqnarray}
Here $\psi$ denotes a spinor component and  the 
real phase $a^{t_j,z_j}_{x_j,y_j}$ is defined as the line integral of the  vector potential ${\bf {\cal A}}$, starting at an arbitrary, but fixed position $(x_o,y_o, z_o)$ and ending on the lattice point 
$(x_j,y_j,z_j)$,
$$
a^{t_j,z_j}_{x_j,y_j}= \frac{q}{\hbar c}\int_{(x_o,y_o,z_o)}^{(x,y,z)} d{\bf s}\cdot {\bf {\cal A}}({\bf s},t)\mid_{x=x_j, y=y_j,z=z_j, t=t_j}~.
$$
$q$ is the fermion electric charge.   Here the grid notation for the more general (3+1)D case has been adopted.  The (2+1)D case can be handled analogously, as discussed in detail for our earlier (2+1)D scheme \cite{hammer2D2cone}.

The implications of the insertion of the Peierls phase-shifted wave functions into the respective numerical schemes  can be summarized as follows \cite{hammer2D2cone}:

%??? okh
\begin{itemize}

\item Since the connectivity of grid points is determined by the structure of the $\alpha-$ or, respectively,  the $\sigma-$ matrices,  requirements on the grid structure are not affected. 

\item Since every spinor component is multiplied merely  by a phase factor, all stability and convergence estimates, as well as the definition of the norm,  can be carried over simply by replacing spinor components by their Peierls transformed $\psi \rightarrow  {\hat\psi}$.

\item The effect of the substitution on difference quotients (derivative terms)

\begin{equation}
\frac{\psi_1-\psi_2}{\Delta} \rightarrow \frac{e^{-ia_1}\psi_1-e^{-ia_2}\psi_2}{\Delta} ~. \nonumber
\end{equation}
can be implemented 
by using the product rule for ``differentiation on the lattice"
\begin{equation}
\frac{e^{-ia_1}\psi_1-e^{-ia_2}\psi_2}{\Delta} =f_+(a_1,a_2)\frac{\psi_1-\psi_2}{\Delta}+\frac{e^{-ia_1}-e^{-ia_2}}{\Delta}\frac{\psi_1+\psi_2}{2}~, \nonumber
\end{equation}
with the definition  $f_\pm(a_1,a_2):=(e^{-ia_1}\pm e^{-ia_2})/2$.
For the case of slow variation of the vector potential over the grid, the  difference quotient for the exponential 
may be approximated by  using the ``chain rule" for the derivative on the grid 
\begin{equation}
\frac{e^{-ia_1}-e^{-ia_2}}{\Delta}= f_+(a_1,a_2)\frac{i(a_2-a_1)}{\Delta} + \frac{1}{\Delta}O((a_1-a_2)^3)~.  \label{chainappx}
\end{equation}

\item  All symmetric averages of the structure
\begin{equation}
\frac{\psi_1+\psi_2}{2} \rightarrow \frac{e^{-ia_1}\psi_1+e^{-ia_2}\psi_2}{2}~. \nonumber
\end{equation}
may be rewritten into
\begin{equation}
\frac{e^{-ia_1}\psi_1+e^{-ia_2}\psi_2}{2}=  f_+(a_1,a_2)\frac{\psi_1+\psi_2}{2} + f_-(a_1,a_2)\frac{\psi_1-\psi_2}{2} \approx  f_+(a_1,a_2)\frac{\psi_1+\psi_2}{2} ~, 
\label{avapp} 
\end{equation}\\
whereby the $f_-$ term may be neglected for vector potentials which are sufficiently smooth on the grid.  
%??? okh
\end{itemize}
~\\
For such smooth vector potentials, approximations Eqs. \eqref{chainappx} and \eqref{avapp} lead us to the following simple rules:
\begin{itemize} 

\item Every partial time-derivative  renormalizes the associated scalar potential term (which occurs in the same equation) according to 

%??? okh
$$
V^{t_j(,z_j)}_{x_j,y_j} \rightarrow V^{t_j(,z_j)}_{x_j,y_j}  -\frac{q}{c}\int_{x_o,y_o(,z_o)}^{x_j,y_j(,z_j)}d{\bf s} \cdot \frac{{\bf {\cal A}}({\bf s},t_j+\Delta_t/2) -{\bf{\cal A}}({\bf s},t_j-\Delta_t/2) }{\Delta_t} ~.
$$

\item Every spatial $j$-derivative term transforms according to 
$$
\frac{\psi_{j+1/2}-\psi_{j-1/2}}{\Delta_j} \rightarrow \frac{\psi_{j+1/2}-\psi_{j-1/2}}{\Delta_j}   -\frac{iq}{\hbar c}{\cal A}_j \frac{\psi_{j+1/2}+\psi_{j-1/2}}{2}~.
$$

\end{itemize}
The substitutions into the schemes above are straight-forward \cite{hammer1D,hammer2D2cone}.  Since the resulting expressions at the various stages of approximation are rather lengthy they are not given here in more detail.
%??? new
\begin{figure}[h!]
\centering
\includegraphics[width=13cm]{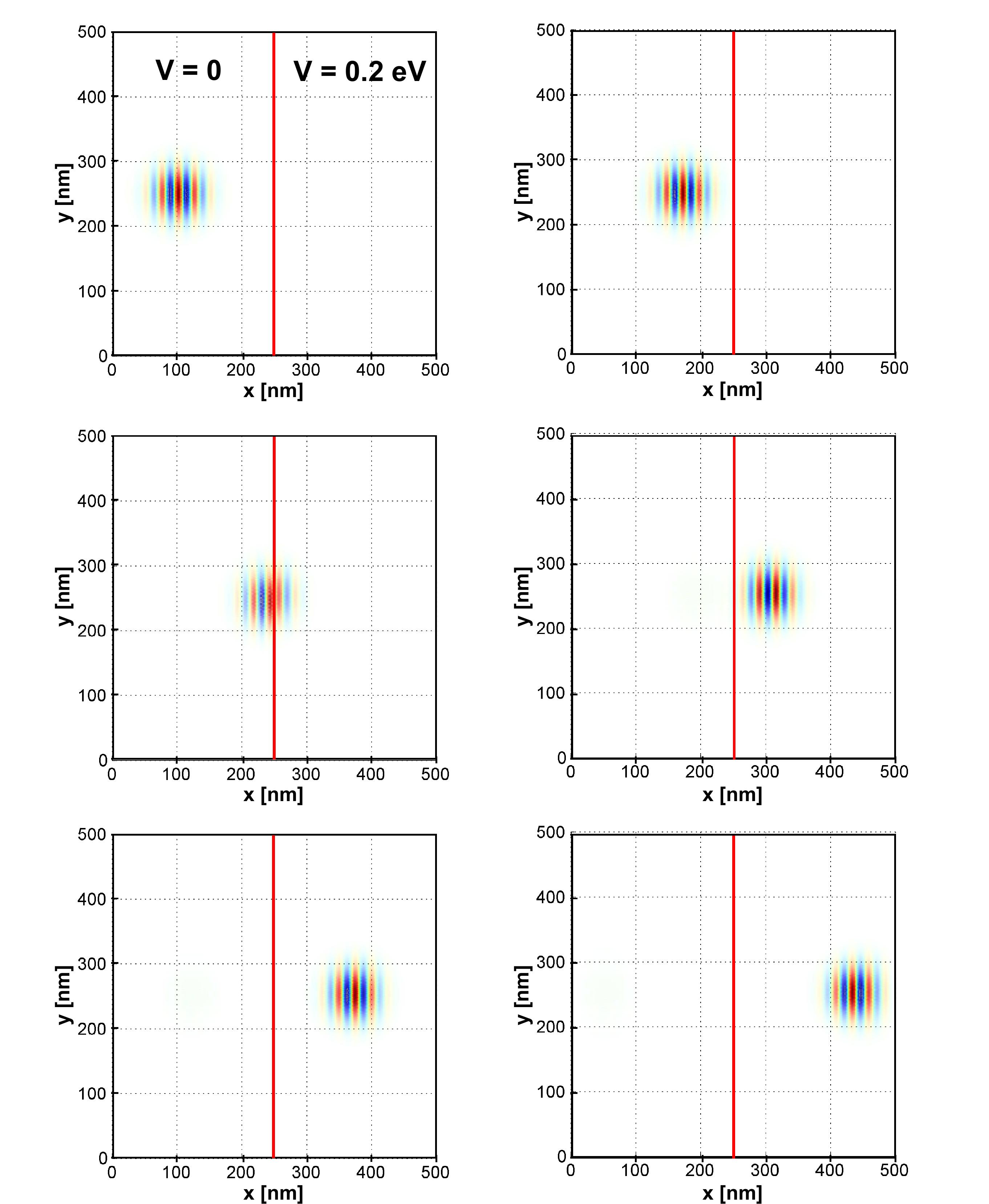}
\caption{Initial Gaussian wave packet impinging on a Klein step in normal direction. %$90^o$ to the vertical. 
%AAAneu
In the six successive snapshots at $t_n = 0,50,100,150,200,250$ the probability density is shown as brightness saturation and the phase is encoded in the color/brightness variation. The mean energy of the wave packet is chosen to be $E=0.1$ eV and the mass is zero $m=0$. At the left-hand side of the step the potential $V=0$ and at the right-hand side it is $V=0.2$ eV, respectively. At the right side the phase velocity and the group velocity are opposite to one another.}\label{ks}
\end{figure}

\begin{figure}[h!]
\centering
\includegraphics[width=13cm]{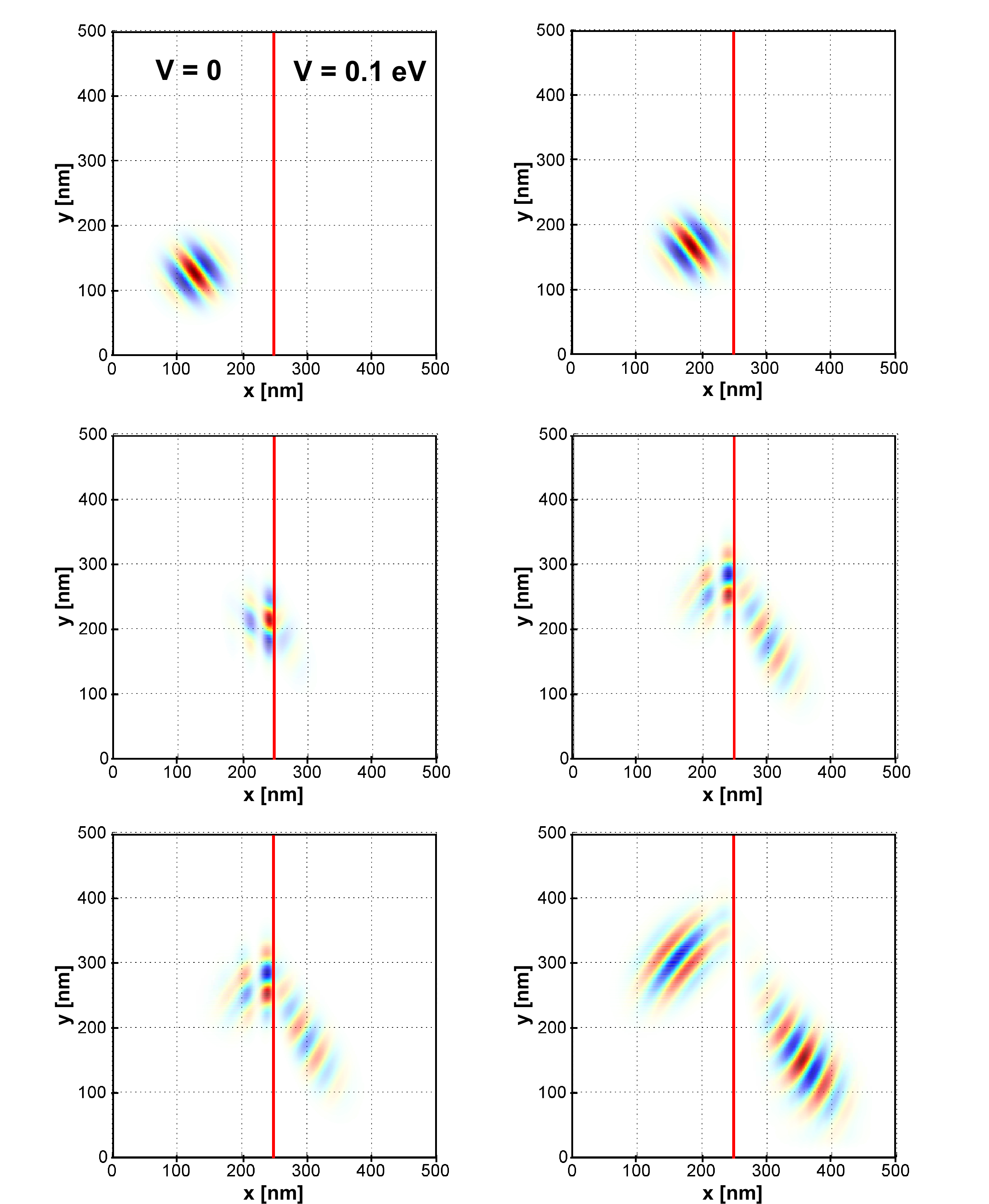}
\caption{Initial Gaussian wave packet impinging on a Klein step at $37^{\circ}$ to the $y$--axis. 
%AAAneu
In the six successive snapshots at $t_n = 0,50,100,150,200,250$ the probability density is shown as brightness saturation and the phase is encoded in the color/brightness variation. The mean energy of the wave packet is chosen to be $E=0.05$ eV and the mass is zero $m=0$. At the left-hand side of the step the potential $V=0$ and at the right-hand side it is $V=0.1$ eV, respectively. At the right side the phase velocity and the group velocity are opposite to one another, leading to a negative angle of refraction.}\label{ksa}
\end{figure}

\section{Simulation of  wave packet scattering at a Klein step}\label{sect5}

Although the (2+1)D scheme has been used  in the simulation of Dirac fermions in electromagnetic textures, here we give an example for illustration \cite{hammerAPL,hammerDW,hammerDISS}.  In particular we consider the scattering of a Gaussian wave packet at a Klein step.  
The latter consists of a step of the  scalar  potential $V$ from zero in the left half-space of Figs. \ref{ks} and \ref{ksa} to a finite value $V>0$ in the right half-space.  We consider the fermion mass $m=0$ and mean energies $E>0$ of the wave packet of less than the step hight.  This causes a resonant energy overlap of  electronic states in the left half-space with hole-states in the right half-space.  In this energy window there is a flip in the sign of the phase velocity  $v_{ph}=\frac{\omega}{k}$ relative to the group velocity  $v_{g}=\frac{\partial\omega}{\partial k}$ ($\omega=E/\hbar$).  In Fig.  \ref{ks} we show an example for normal incidence of the wave packet onto the potential step, while Fig.  \ref{ksa} shows snapshots of the propagation for incidence under  $37^{\circ}$ to the $y$--axis.  
%AAAneu
The simulation shows the nearly complete penetration of the wave packet into and through  the barrier under normal incidence (Klein paradox) due to the choice of parameters.   For general angle of incidence one has 
both 
a reflected and transmitted wave packet whereby the latter displays a negative index of refraction, as evident in Fig.  \ref{ksa} and is the characteristic of a meta-material.    We note the faithful representation of the propagation direction of both the wave packet and the phase (indicated in color) due to the absence of spurious additional cones in spite of the presence of an abrupt potential step.  Absorbing boundary layers have been used to suppress spurious backscattering from the simulation boundaries \cite{hammerDW,hammerDISS}.
%??? new
\section{Summary and Conclusions}\label{suco}

We have presented a real-space finite-difference method for both the four-component (3+1)D  and  the two-component (2+1)D Dirac equation. In both cases, the scheme is capable of handling general 
time-dependent electromagnetic potential and mass terms. Stability is proven for this general case and an exactly conserved functional is identified.  It can be interpreted as a probability density. 

The proposed schemes share and combine several decisive advantages over previous methods.  
Most importantly, the fermion doubling problem is avoided by staggering the spinor components symmetrically both in space and time. The strictly monotonic dispersion relation features a single Dirac cone. The scheme is second order accurate and, because it is explicit, shows an optimal linear scaling behavior.  It lends itself to a convenient introduction of absorbing boundary conditions via boundary regions with imaginary scalar potential contributions, following a strategy discussed elsewhere \cite{hammer2D2cone}.  Furthermore, this scheme which we have developed explicitly for one, two, and three spatial dimensions can be extended to higher dimensions.  

The code readily is parallelized. This can be achieved, for example,  by dividing the problem into spatial sub-domains, each computed on one CPU and communicating only the boundary grid-points. 
Although Maxwell's equations have a different structure, the method presented here for the Dirac equation has much in common with the finite difference time domain (FDTD) method.   Today it is applied with great success in computational electromagnetics \cite{yee,taflove}. This close relationship also allows for the use of tools, initially developed for the FDTD time domain method, like the perfectly matched layer (PML), in its various variants,  acting as an absorbing (open) boundary for finite-domain simulations \cite{berenger,johnson}.\\

Since the presented schemes allow for the incorporation of the full electromagnetic four-vector potential, self-consistent dynamics (e.g.,  self-energy corrections) mediated by the electromagnetic interaction is readily 
achieved by solving the electromagnetic potential equations in parallel.  This allows the utilization of all the advantages associated with the present approach and eliminates the need for the discretization of non-linear extensions of the Dirac equation.  On the other hand, the present method may prove to be useful also for the discretization of such nonlinear Dirac equations, in particular, where nonlinearities arise from other types of interaction \cite{xu:2013}.
%???  okh

 In summary, the combination of these favorable properties makes this approach highly suitable for the numerical treatment of Dirac fermion dynamics in space- 
and time-dependent external fields.  As such, it should be useful to a variety of fields in physics, ranging from elementary particle, atomic, molecular, to condensed matter and astrophysics.
%??? okh

~\\

\noindent{\bf Acknowledgment:}  This work has been supported by the Austrian Science Foundation under Project No. I395-N16.\\

\bigskip

%%%%%%%%%%%%%%%%%%%%%%%%%%%%%%%%%%%%%%%%%%%%%%%%%%%%%%%%%%%%%%%%%%%%%%%%%%%%%%%%%%%%%%%%%%
%-------APPENDIX-------
%
%\newpage

\begin{appendix}
\section{Stability proof for the (3+1)D leap-frog staggered-grid scheme\label{A}}
%AAAneu
\noindent In this appendix we provide the proof of Proposition \ref{Prop3}. To this end we shall need the conserved functional \eqref{EPM} for the case $r=r_x=r_y=r_z= 1/\sqrt{3}$:
%neuAAA
\begin{align}
E =& \sum_{j}  |A_{j}|^2 + |B_{j+(\frac{1}{2},0,\frac{1}{2})}|^2 + |C_{j+(0,0,\frac{1}{2})}|^2 + |D_{j+(\frac{1}{2},0,0)}|^2\\\nonumber
&- \frac{1}{\sqrt{3}} \Re \bigg[\Big(C_{j+(0,0,\frac{1}{2})}-C_{j-(0,0,\frac{1}{2})}+D_{j+(\frac{1}{2},0,0)}-D_{j-(\frac{1}{2},0,0)} - i D_{j+(0,\frac{1}{2},0)} + i D_{j-(0,\frac{1}{2},0)},A_j
\Big)\bigg]\\\nonumber
&- \frac{1}{\sqrt{3}}
\Re\bigg[\Big(-D_{j+(\frac{1}{2},0,1)}+D_{j+(\frac{1}{2},0,0)}+C_{j+(1,0,\frac{1
}{2})}-C_{j+(0,0,\frac{1}{2})} + i C_{j+(\frac{1}{2},\frac{1}{2},\frac{1}{2})} -
i C_{j+(\frac{1}{2},-\frac{1}{2},\frac{1}{2})},B_{j+(\frac{1}{2},0,\frac{1}{2})}
\Big)\bigg]~,
\end{align} 
with $j\in \mathbb{Z}^3 \cup (\mathbb{Z}+1/2)^2\times\mathbb{Z}$ in the sum.
\begin{proof}[Proof of Proposition \ref{Prop3}]
Let us first define the auxiliary quantity
\begin{align}
\tilde{E} :&= \frac{1}{2} \sum_{j} \bigg|\frac{
C_{j+(0,0,\frac{1}{2})}+D_{j+(\frac{1}{2},0,0)}- i D_{j+(0,\frac{1}{2},0)}}{\sqrt{3}}-A_j\bigg|^2\\\nonumber %
&+ \frac{1}{2} \sum_{j} \bigg|\frac{
C_{j-(0,0,\frac{1}{2})}+D_{j-(\frac{1}{2},0,0)}- i D_{j-(0,\frac{1}{2},0)}}{\sqrt{3}}+A_j\bigg|^2\\\nonumber %
&+ \frac{1}{2} \sum_{j} \bigg|\frac{
D_{j+(\frac{1}{2},0,1)}-C_{j+(1,0,\frac{1}{2})}- i C_{j+(\frac{1}{2},\frac{1}{2},\frac{1}{2})}}{\sqrt{3}}+B_{j+(\frac{1}{2},0,\frac{1}{2})}\bigg|^2\\\nonumber %
&+ \frac{1}{2} \sum_{j} \bigg|\frac{
D_{j+(\frac{1}{2},0,0)}-C_{j+(0,0,\frac{1}{2})}- i C_{j+(\frac{1}{2},-\frac{1}{2},\frac{1}{2})}}{\sqrt{3}}-B_{j+(\frac{1}{2},0,\frac{1}{2})}\bigg|^2~.
\end{align}
It satisfies
\begin{align}
\tilde{E} &=\sum_{j}  |A_{j}|^2 + |B_{j+(\frac{1}{2},0,\frac{1}{2})}|^2 + |C_{j+(0,0,\frac{1}{2})}|^2 + |D_{j+(\frac{1}{2},0,0)}|^2\\\nonumber
&~~+\frac{1}{2\sqrt{3}} \sum_{j} \bigg[
-\big(C_{j+(0,0,\frac{1}{2})}+D_{j+(\frac{1}{2},0,0)}- i
D_{j+(0,\frac{1}{2},0)}\big)A_j^*\\\nonumber 
&\qquad\qquad\quad~-\big(C_{j+(0,0,\frac{1}{2})}+D_{j+(\frac{1}{2},0,0)}- i D_{j+(0,\frac{1}{2},0)}\big)^*A_j\\\nonumber 
&\qquad\qquad\quad~+\big(C_{j-(0,0,\frac{1}{2})}+D_{j-(\frac{1}{2},0,0)}- i D_{j-(0,\frac{1}{2},0)}\big)A_j^*\\\nonumber 
&\qquad\qquad\quad~+\big(C_{j-(0,0,\frac{1}{2})}+D_{j-(\frac{1}{2},0,0)}- i D_{j-(0,\frac{1}{2},0)}\big)^*A_j\\\nonumber
&\qquad\qquad\quad~+\big(D_{j+(\frac{1}{2},0,1)}-C_{j+(1,0,\frac{1}{2})}- i C_{j+(\frac{1}{2},\frac{1}{2},\frac{1}{2})})B_{j+(\frac{1}{2},0,\frac{1}{2})}^*\\\nonumber
&\qquad\qquad\quad~+\big(D_{j+(\frac{1}{2},0,1)}-C_{j+(1,0,\frac{1}{2})}- i C_{j+(\frac{1}{2},\frac{1}{2},\frac{1}{2})})^*B_{j+(\frac{1}{2},0,\frac{1}{2})}\\\nonumber
&\qquad\qquad\quad~-\big(D_{j+(\frac{1}{2},0,0)}-C_{j+(0,0,\frac{1}{2})}- i
C_{j+(\frac{1}{2},-\frac{1}{2},\frac{1}{2})})B_{j+(\frac{1}{2},0,\frac{1}{2})}
^*\\\nonumber
&\qquad\qquad\quad~-\big(D_{j+(\frac{1}{2},0,0)}-C_{j+(0,0,\frac{1}{2})}- i
C_{j+(\frac{1}{2},-\frac{1}{2},\frac{1}{2})})^*B_{j+(\frac{1}{2},0,\frac{1}{2})}
\bigg]\\\nonumber
&~~+\frac{1}{6} \sum_{j} \bigg[
C_{j+(0,0,\frac{1}{2})}\big(D_{j+(\frac{1}{2},0,0)}^*+i D_{j+(0,\frac{1}{2},0)}^*\big)
+C_{j+(0,0,\frac{1}{2})}^*\big(D_{j+(\frac{1}{2},0,0)}-i
D_{j+(0,\frac{1}{2},0)}\big)\\\nonumber
& \qquad\qquad
+C_{j-(0,0,\frac{1}{2})}\big(D_{j-(\frac{1}{2},0,0)}^*+i D_{j-(0,\frac{1}{2},0)}^*\big)
+C_{j-(0,0,\frac{1}{2})}^*\big(D_{j-(\frac{1}{2},0,0)}-i
D_{j-(0,\frac{1}{2},0)}\big)\\\nonumber
& \qquad\qquad
+D_{j+(\frac{1}{2},0,1)}\big(-C_{j+(1,0,\frac{1}{2})}^* + i C_{j+(\frac{1}{2},\frac{1}{2},\frac{1}{2})}^*\big)
+D_{j+(\frac{1}{2},0,1)}^*\big(-C_{j+(1,0,\frac{1}{2})} - i
C_{j+(\frac{1}{2},\frac{1}{2},\frac{1}{2})}\big)\\\nonumber
& \qquad\qquad
+D_{j+(\frac{1}{2},0,0)}\big(-C_{j+(0,0,\frac{1}{2})}^* + i C_{j+(\frac{1}{2},-\frac{1}{2},\frac{1}{2})}^*\big)
+D_{j+(\frac{1}{2},0,0)}^*\big(-C_{j+(0,0,\frac{1}{2})} - i
C_{j+(\frac{1}{2},-\frac{1}{2},\frac{1}{2})}\big)\bigg]\\\nonumber
&~~+\frac{i}{6} \sum_{j} \bigg[
D_{j+(\frac{1}{2},0,0)} D_{j+(0,\frac{1}{2},0)}^*-D_{j+(\frac{1}{2},0,0)}^* D_{j+(0,\frac{1}{2},0)}
+D_{j-(\frac{1}{2},0,0)} D_{j-(0,\frac{1}{2},0)}^*-D_{j-(\frac{1}{2},0,0)}^* D_{j-(0,\frac{1}{2},0)}\bigg]\\\nonumber
&~~+\frac{i}{6} \sum_{j} \bigg[
-C_{j+(1,0,\frac{1}{2})} C_{j+(\frac{1}{2},\frac{1}{2},\frac{1}{2})}^*
+C_{j+(1,0,\frac{1}{2})}^* C_{j+(\frac{1}{2},\frac{1}{2},\frac{1}{2})}
-C_{j+(0,0,\frac{1}{2})} C_{j+(\frac{1}{2},-\frac{1}{2},\frac{1}{2})}^*
+C_{j+(0,0,\frac{1}{2})}^* C_{j+(\frac{1}{2},-\frac{1}{2},\frac{1}{2})}\bigg]\\\nonumber
&=E~,
\end{align}
since the last and the penultimate sum vanish because they are telescopic sums:
In both sums the first (resp.~second) term with $j=(0,0,0)$ cancels the
forth (resp.~third) term with $j=(\frac{1}{2},\frac{1}{2},0)$. In the
antepenultimate sum there occurs direct cancelation in two pairs of terms
(involving $C_{j+(0,0,\frac{1}{2})} D_{j+(\frac{1}{2},0,0)}^*$ and
$C_{j+(0,0,\frac{1}{2})}^* D_{j+(\frac{1}{2},0,0)}$) and the remaining terms
are again telescopic sums.\\
\\
Using $\frac{1}{4}\left\|a_1 + a_2\right\|^2 \leq \frac{1}{4} \big(\left\|a_1 + b\right\|+\left\|a_2 - b\right\|\big)^2\leq \frac{1}{2} \left\|a_1 + b\right\|^2 + \frac{1}{2} \left\|a_2 - b\right\|^2~$ %AAAneu
gives for the grid-averaged norm (defined in \eqref{CD-norm})
$$
  \|(\tilde{C},\tilde{D})\|^2\leq\tilde{E} = E~.
$$
The symmetry of the scheme and of $E$ w.r.t.~$(C,D)$ and $(A,B)$ yields also:
$\|(\tilde{A},\tilde{B})\|^2\leq\tilde{E} = E$.
%neuAAA
%
%
\end{proof}
\end{appendix}

%%%%%%%%%%%%%%%%%%%%%%%%%%%%%%%%%%%%%%%%%%%%%%%%%%%%%%%%%%%%%%%%%%%%%%%%%%%%%%%%%%%

%\newpage

\end{document}